\documentclass[11pt]{article}
\usepackage{color}
\usepackage[total={6.5in,8.75in}, top=1.2in, left=0.9in, includefoot]{geometry}
\usepackage{graphicx}
\usepackage{amsmath, amsthm, amssymb}
\usepackage{url}
\usepackage{algorithmic}
\usepackage{xspace}  

\newtheorem{theorem}{Theorem} 

\newtheorem{lemma}[theorem]{Lemma}
\newtheorem{algorithm}[theorem]{Algorithm} 

\theoremstyle{definition}
\newtheorem*{definition}{Definition}

\newcommand{\Eq}[1]{(\ref{eq:#1})}
\newcommand{\Thm}[1]{Thm.~\ref{thm:#1}}
\newcommand{\Lem}[1]{Lem.~\ref{lem:#1}}

\newcommand{\Sec}[1]{\S \ref{sec:#1}}
\newcommand{\Fig}[1]{Fig.~\ref{fig:#1}}

\newcommand{\App}[1]{App.~\ref{app:#1}}
\newcommand{\Alg}[1]{Alg.~\ref{alg:#1}}

\newcommand{\eps}{\varepsilon}

\newcommand{\bN}{{\mathbb{ N}}}

\newcommand{\bR}{{\mathbb{ R}}}

\newcommand{\bZ}{{\mathbb{ Z}}}

\newcommand{\cA}{\ensuremath{\mathcal{A}}\xspace}
\newcommand{\cC}{\ensuremath{\mathcal{C}}\xspace}
\newcommand{\cD}{\ensuremath{\mathcal{D}}\xspace}
\newcommand{\cF}{\ensuremath{\mathcal{F}}\xspace}

\newcommand{\cK}{\ensuremath{\mathcal{K}}\xspace}
\newcommand{\cL}{\ensuremath{\mathcal{L}}\xspace}

\newcommand{\cR}{\ensuremath{\mathcal{R}}\xspace}
\newcommand{\cV}{\ensuremath{\mathcal{V}}\xspace}
\newcommand{\cW}{\ensuremath{\mathcal{W}}\xspace}

\newcommand{\inv}{\text{inv}}
\newcommand{\cl}{\text{cl}}
\newcommand{\Int}{\text{int}}
\newcommand{\simplex}[1]{ \langle {#1} \rangle}

\newcommand{\FA}{\ensuremath{F_\cA}\xspace}
\newcommand{\FK}{\ensuremath{F_{\cK}}\xspace}
\newcommand{\FW}{\ensuremath{F_{\cW}}\xspace}
\newcommand{\cFK}{\ensuremath{\cF_\cK}\xspace}
\newcommand{\cFW}{\ensuremath{\mathcal{F_W}}\xspace}
\newcommand{\Conley}{\operatorname{Con}}

\newcommand{\beq}[1]{\begin{equation}\label{eq:#1}}
\newcommand{\eeq}{\end{equation}}

\newenvironment{se}[1]{\equation\label{eq:#1}\aligned}{\endaligned\endequation}
\newcommand{\bsplit}[1]{\begin{se}{#1}}
\newcommand{\esplit}{\end{se}}

\newcommand{\InsertFig}[4]
{\begin{figure}[h!t]
       \centerline{
         \includegraphics[width=#4]{./figures/#1}
       }
       \caption{{\footnotesize  #2}
       \label{fig:#3}}
\end{figure}}

\title{Simplicial Multivalued Maps and the Witness Complex for Dynamical Analysis of Time Series}

 \author{
   Zachary Alexander, Elizabeth Bradley, James D. Meiss, and Nicole Sanderson    
\\\\
\begin{tabular}{cc}
	Departments of Computer Science, Applied Mathematics \& Mathematics\\
    University of Colorado \\
	Boulder, CO 80309-0526 \\
\end{tabular}
}

\begin{document}

\maketitle

\begin{abstract}

Topology based analysis of time-series data from dynamical systems is
powerful: it potentially allows for computer-based proofs of the
existence of various classes of regular and chaotic invariant sets for
high-dimensional dynamics.  Standard methods are based on a cubical
discretization of the dynamics and use the time series to construct
an outer approximation of the underlying dynamical system. The
resulting multivalued map can be used to compute the Conley index of
isolated invariant sets of cubes. In this paper we introduce a
discretization that uses instead a simplicial complex
constructed from a witness-landmark relationship. The goal is to
obtain a natural discretization that is more tightly connected with
the invariant density of the time series itself. The time-ordering of
the data also directly leads to a map on this simplicial complex that
we call the witness map. We obtain conditions under which this witness
map gives an outer approximation of the dynamics, and thus can be used
to compute the Conley index of isolated invariant sets. The method is
illustrated by a simple example using data from the classical H\'enon
map.

\end{abstract}

\bigskip

\section{Introduction}\label{sec:intro}

Our goal in this paper is to develop some new computational topology techniques to characterize
some aspects of discrete or continuous dynamical systems.  We assume that the only
knowledge we have of the dynamics is a finite time series $\Gamma =
\{x_0,x_1,\ldots, x_{T-1}\}$ taken from a state-space trajectory of
the system.  If the system is a map, $f: X \to X$, then
$\Gamma$ is simply the iterates of the map: $x_{t+1}=f(x_t)$.  If the
system is a flow, then $\Gamma$ is a sequence of samples of the
continuous trajectory $x(t)$ and we are effectively studying the
evolution operator that maps the system forward in time.  In either
case, given this information, we cannot hope to approximate the
dynamics on all of $X$; instead, we assume that $\Gamma$ lies close to $\Lambda$,
a bounded invariant set of $f$.
For example, any orbit in the basin of an attractor will eventually approach it,
so in this case, $\Gamma$ can be taken to be the trajectory after a transient is removed.
Thus our goal is to develop tools that will allow us to characterize
properties of  $f|_\Lambda$, such as the 
number and types of periodic orbits, compute topological entropy, etc. 

The tool that we use is the discrete Conley index (we recall the definition of this index and related concepts in \App{Conley})  \cite{Conley78,Easton98,KMM04}.  The Conley index characterizes some dynamical properties of isolated invariant sets of $f$---for example it may establish the existence of periodic orbits or give a lower bound on the topological entropy. 

Since we only have access to the trajectory $\Gamma$, we must use it to obtain an approximation of the underlying dynamics before we can compute the Conley index. There are two categories of maps that can serve this purpose. The first, as we recall in \Sec{multivalued-maps}, are multivalued maps. These are set-valued and typically defined on a finite covering of a neighborhood of the invariant set $\Lambda$; 
they capture how the images of the cover map across other elements of the cover. Multivalued maps have most commonly been defined on cubical grids \cite{Mischaikow97, Mischaikow99}, but as we discuss in \Sec{Grids}, more-general grids can also be used. A multivalued map on a grid, which we call a \emph{cellular multivalued map} (CMM) in \Sec{cellular:maps}, is defined to be constant on the interior of each cell as well as on subsets of the boundary where groups of cells intersect. Note that while connectivity is obvious in uniform grids of cubical cells, this may not be the case in other grid geometries. The second category of map addresses this issue. Dual to any grid, cubical or otherwise, is a simplicial complex---the \emph{nerve} of the grid (see \App{Complexes}). We call the map induced on this complex a \emph{simplicial multivalued map} (SMM) in \Sec{cellular:maps}. 
When the number of grid cells is finite, such maps are finitely representable: they can be stored precisely in a computer and used to perform exact computations. 

The Conley index for an isolated invariant set (defined in \App{Conley}) can be computed using a corresponding CMM or SMM (techniques are recalled in \App{ComputingConley}). 
Moreover, as described in \Sec{cellular:maps},
if the cellular map is semicontinuous and \emph{acyclic}, then the map that it
induces on homology coincides with that induced by $f$, and thus it can be used to
compute the Conley index of isolated invariant sets of $f$.

The associated computational cost of these computations depends on the geometry of the cells. To minimize this complexity while still preserving the essential features, our approach uses two constructions that play major roles in computational topology: the \emph{$\alpha$-diagram} \cite{Edelsbrunner92,Edelsbrunner95} and the \emph{witness complex} \cite{deSilva04}. In \Sec{AlphaComplex} we recall that the $\alpha$-diagram of a data set is the intersection of its Voronoi-diagram  with the union of balls of radius $\alpha$ centered on the data points.  The nerve of the $\alpha$-diagram is the $\alpha$-complex: it is generically a simplicial complex and is a subset of the Delaunay triangulation (see \App{Complexes}),  limiting to the latter as $\alpha \to \infty$ \cite{Edelsbrunner95}. Since the geometry of cells in the $\alpha$-diagram is dictated by the data, rather than by rectilinear grid lines, the shape of the $\alpha$-complex naturally follows that of the invariant set $\Lambda$.

While this flexible, data-driven representation has some appealing
advantages, an $\alpha$-complex constructed from a long time series $\Gamma =
\{x_0,x_1,\ldots, x_{T-1}\}$  would have at least one simplex for each point, and 
the complexity of algorithms that construct and manipulate these
objects scales poorly with the number of simplices.
It is useful, then, to represent these data using a global topological
object that contains fewer simplices while preserving the Conley
index.  We use the \emph{witness complex} \cite{deSilva04} for this
purpose, see \Sec{witness:map}.  Instead of assigning a vertex to each point
in $\Gamma$, we represent the data by a smaller set of vertices, a set
of \emph{landmarks}, $L \subset X$, and build a simplicial complex
from those points.  As described in \Sec{AlphaComplex}, there are a
number of ways to choose landmarks. The computational complexity of this approach and
its comparison to that of a cubical-grid are discussed in
\App{Complexity}. 

The witness complex is constructed from a relation on $\Gamma \times L$: 
each point in $\Gamma$ may be a witness to one or
more landmarks and each landmark may have one or more witnesses.
Of the many possible definitions of witness relation, we choose one in which a point
$x \in \Gamma$ witnesses a set $\sigma \subset L$ if the distance between
$x$ and any landmark in $\sigma$ is no more than $\eps$ greater
than the minimum distance between $x$ and the full set $L$ of landmarks,
see \Sec{WitnessComplex}.
We use this witness relation to construct an abstract
witness complex. The simplest
implementation gives a clique or flag complex: it consists of
simplices whose pairs of vertices have a common witness.  
As we show in \Sec{equivalence}, if the landmarks are selected to
be sufficiently uniform and the trajectory is sufficiently dense,
then there are conditions under which the witness complex and the 
$\alpha$-complex are the same.

The dynamics on the time series induces a simplicial multivalued map on the witness complex. 
This SMM also induces a corresponding cellular multivalued map on a grid of $\alpha$-cells based on the landmarks. These \emph{witness maps}, which are a primary contribution of
this paper, are described in \Sec{WitnessMap}.  The construction of
the witness map is developed in several steps in \Sec{witness:map} in order to bring the
well-developed theory of \cite{KMM04} to bear upon this new
formulation and thereby establish the correspondence between the
homology of the witness map and that of the true dynamics of the
underlying system. 

Finally, in \Sec{example} we use data generated from the classic H\'enon quadratic map to give a simple illustration of the ideas in this paper. We show that, under a verifiable set of assumptions, our techniques could be used to obtain rigorous results about the underlying dynamical system.


\section{Multivalued Maps} \label{sec:multivalued-maps}

In this section we describe the concept of \emph{multivalued maps} and
obtain criteria that imply for such maps to be an \emph{enclosure} of a
map $f$. A \emph{cellular multivalued map} is defined to be constant
on each cell of a \emph{grid}, generalizing the cubical case of
\cite{KMM04}. A cellular map gives rise to a \emph{simplicial
  multivalued map} on the nerve of the grid. In the final part of this
section we show that one way to construct a grid is through an
\emph{$\alpha$-diagram}. In this case, the nerve is a geometrical
simplicial complex that is a deformation retract of the grid, and we
will show that the cellular map and the simplicial map induce the same
maps on homology.

Since we are interested in applications to data sets that
correspond to real-valued measurements of continuous dynamical
systems, we will assume that our time series $\Gamma$ is obtained from
a map $f$ on a submanifold $X \subset \bR^n$.  We will use the
Euclidean metric, $d(\cdot,\cdot)$, on $\bR^n$.  In the future, it might be 
useful to consider more-general metrics on the submanifold $X$ itself.

We begin by recalling some standard definitions for
multivalued maps that approximate a dynamical system.\footnote
{Following \cite{Day04, Day08}, we denote single-valued maps with lower-case letters (e.g., $f$), sets and set-valued maps with capital letters, e.g., $F$, and combinatorial objects and maps with calligraphic letters, e.g. \cF.}

\begin{definition}[Multivalued Map]
	A \emph{multivalued map}, $F:X \rightrightarrows X$, is a map from
    $X$ to its power set. That is, for each $x \in X$, $F(x)$ is a subset of $X$.\end{definition}

We use multivalued maps to approximate continuous maps $f:X \to X$,
and the approximation is taken to be ``good" if the action on homology 
induced by $F$ is
equivalent to that induced by $f$. In order for this to be the case, the
action of $F$ must
\emph{enclose} the action of $f$ and not introduce any extra
homological structure. These requirements are spelled out in the
following definitions.  

\begin{definition}[Outer Approximation]\label{def:outer:approximation}
	A multivalued map $F:X \rightrightarrows X$ is an \emph{outer approximation} of a continuous map $f:X \to X$ if $f(x) \in F(x)$ for each $x \in X$. In this case $f$ is said to be a \emph{continuous selector} for $F$.
\end{definition}

\noindent
The (weak) preimage of a multivalued map is itself a multivalued map defined as
\[
	F^{-1}(y) = \{ x \in X :  y \in F(x) \}.
\]
\begin{definition}[Semicontinuous]\label{def:semi:continuous}
	A multivalued map is (lower) \emph{semicontinuous} if the preimage of each open set is open. 
\end{definition}

\noindent
As usual, the $n$-dimensional homology group of a set $X$ is denoted $H_{n}(X)$. We use, for simplicity, the homology over $\bZ_2$ so that the torsion subgroups are ignored. Given this, a multivalued map preserves homology if the image of each point is homologous to a point. This is captured in the following definition.
\begin{definition}[Acyclic] \label{def:acyclic}
	A multivalued map $F$ is \emph{acyclic} if for each $x \in X$,
\[
		H_{n}(F(x)) = \left\{\begin{array}{cc} 
		                     \bZ_2\ ,\ &n = 0 \\ 
						    \mathbf{0}\ ,\ &n > 0 
					   \end{array} \right. .
\]
\end{definition}

\noindent
The key point is that if $F$ is semicontinuous and acyclic, then every continuous selector for $F$ induces the same homomorphism in homology---a consequence of the acyclic carrier
theorem \cite[Thm.~13.3]{Munkres84}. 

\begin{definition}[Enclosure] A semicontinuous, acyclic, multivalued map $F$ is an \emph{enclosure} of any continuous selector $f$.
\end{definition}

\noindent
Therefore, one can define the homology induced by an enclosure to be that of any of its continuous selectors.

In order that this homology be computable, however, it is necessary to obtain a finitely representable approximation of the map $f$, and it is this to which we turn next.

\subsection{Grids} \label{sec:Grids}

A grid allows one to construct finitely representable maps that can be
outer approximations of a map $f$ \cite{Kalies05, Mrozek99}. We will
consider generalizations of the cubical cells of \cite{KMM04} to a
grid constructed from a collection of cells $\cA = \{A_1,A_2, \dots\}$
in $X$. 
Associated with any such collection is its geometrical realization---the union of these cells as subsets of $X$---denoted by 
\beq{GeometricalRealization} 
	|\cA|:= \bigcup_{A \in \cA} A.  
\eeq 
Since the shape and number of neighbors of each cell can
vary, such a grid may permit more-efficient computational algorithms than those 
for cubes of fixed size and shape.\footnote
{See \App{Complexity} for more discussion of algorithms and associated computational complexity issues.}
Nevertheless, many of the results that appear here are easily adapted from the cubical case.

There are four basic properties for the cells of a grid.

\begin{definition}[Grid \cite{Arai09}] \label{def:abstract:grid}
	A family of nonempty compact sets \cA is a \emph{grid} on $X$ if
	\begin{enumerate}
		\item[a)] $X = |\cA|$;
		\item[b)] for all $A \in \cA $, $ A = \text{cl}(\Int(A))$;
		\item[c)] for all $A,B \in \cA$ if $A \neq B $ then $ A \cap \Int(B) = \emptyset$; and
		\item[d)] a finite subset of $\cA$ covers each compact $S \subset X$.
	\end{enumerate}
\end{definition}
\noindent
A prototypical grid is a lattice of closed cubes; indeed, this is the example studied extensively in \cite{KMM04}. An example of a more-general grid is shown in \Fig{CellularMap}. The cells in this grid are $\alpha$-cells, see \Sec{AlphaComplex}.

\InsertFig{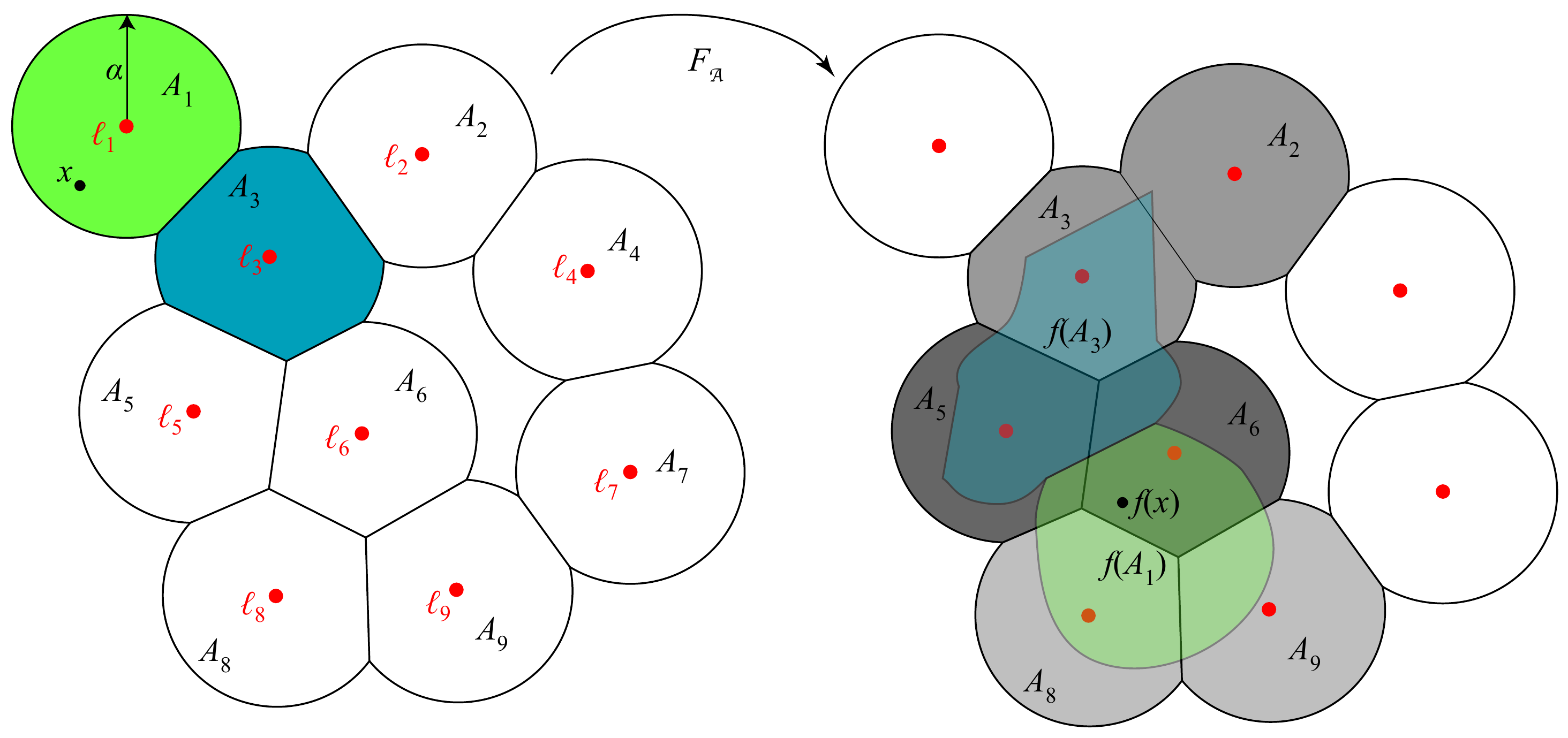}{Sketch of a a grid, $\cA_\alpha(L)$, of $\alpha$-cells (see \Sec{AlphaComplex}) based on a set of nine landmarks $\{l_i\}$ (red points) and the action of a cellular map $\FA$ \Eq{CellularMap} on two of these cells. Here the image of any $x \in \Int(A_1)$ is $A_5 \cup A_6 \cup A_8 \cup A_9$ and of any 
  $x \in \Int(A_3)$ is $A_2 \cup A_3 \cup A_5 \cup A_6$. 
  If $x \in A_1 \cap A_3$ then $\FA(x) = A_5 \cup
  A_6$.}{CellularMap}{6in}

\subsection{Cellular and Simplicial Maps}\label{sec:cellular:maps}
A cellular multivalued map is a map on the geometrical realization,
\Eq{GeometricalRealization}, of a grid $\cA$ that is constant on the interior of 
each cell of $\cA$. It generalizes the cubical multivalued map of \cite{KMM04} to a 
situation in which the cell boundaries need not be rectilinear.

\begin{definition}[Cellular Multivalued Map (CMM)]
A multivalued map $\FA: |\cA| \rightrightarrows |\cA|$ on the geometrical realization  of a grid \cA is a cellular multivalued map if it is the outer approximation of $f$ defined by
\beq{CellularMap}
	\FA(x) := 
	            \bigcap_{B\in \cA\,:\,x \in B }\{ A \in \cA: A \cap f(B) \neq \emptyset\} . 
\eeq
\end{definition}

\noindent
The map \FA takes the interior of each cell to the
union of the cells that intersect its image and
the boundary shared by multiple cells to the cells that contain the intersection of their
images.  An illustration is shown in \Fig{CellularMap}.
An implication is that CMMs are
semicontinuous because they map the boundary of each cell to a subset
of the image of the cell itself---this is a straightforward generalization of
\cite[Prop. 6.17]{KMM04} for the related cubical case.

This construction is not easy to implement on a computer for two reasons:
it is a map on a continuum $|\cA|$, and to construct it we must
know $f(x)$ for each point $x \in |\cA|$.  The second problem is
addressed by the \emph{witness map} introduced in \Sec{WitnessMap}.
The first problem can be mitigated by defining a finite map, $\cF: \cK
\rightrightarrows \cK$, on a complex $\cK$ related to the grid
$\cA$. One such complex is the CW-complex formed from $\cA$: that is,
the collection of cells $\{A_i\}$, together with their faces $\{A_i \cap A_j\}$, their edges, and so forth.  An example is
the cubical CW-complex used in the approach of \cite{KMM04}. Since the
cellular multivalued map \FA is constant on each cell of the CW-complex, it naturally gives rise to an associated combinatorial map. For example
in \Fig{CellularMap} the two-cell $A_1$ would be mapped to
$\{A_5,A_6,A_8, A_9\}$ and the one-cell represented by $A_1\cap A_3$
would be mapped to $\{A_5,A_6\}$.

Our goal is to use a more easily described complex that is also
naturally suited to homology calculations: in particular, a simplicial
complex.
A natural simplicial complex to use is the \emph{nerve} $N(\cA)$, recall \App{Complexes}. Let $L$ denote a set of labels for the elements of $\cA$, and
$\sigma = \simplex{l_0,l_2,\ldots, l_k}$ be any finite subset of $L$. The intersection of the cells labeled by the elements of $\sigma$ is denoted
\beq{Atau} 
	A_\sigma := \bigcap_{l \in \sigma} A_{l} .
\eeq
A set $\sigma$ is a simplex in the nerve when $A_\sigma \neq \emptyset$, thus 
\beq{NerveComplex}
	\cK := N(\cA) = \{ \sigma = \simplex{l_0,l_1,\ldots, l_k}: A_\sigma \neq \emptyset\} .
\eeq
For example, the nerve of the grid in \Fig{CellularMap} is shown in \Fig{SimplicialMap}.

\InsertFig{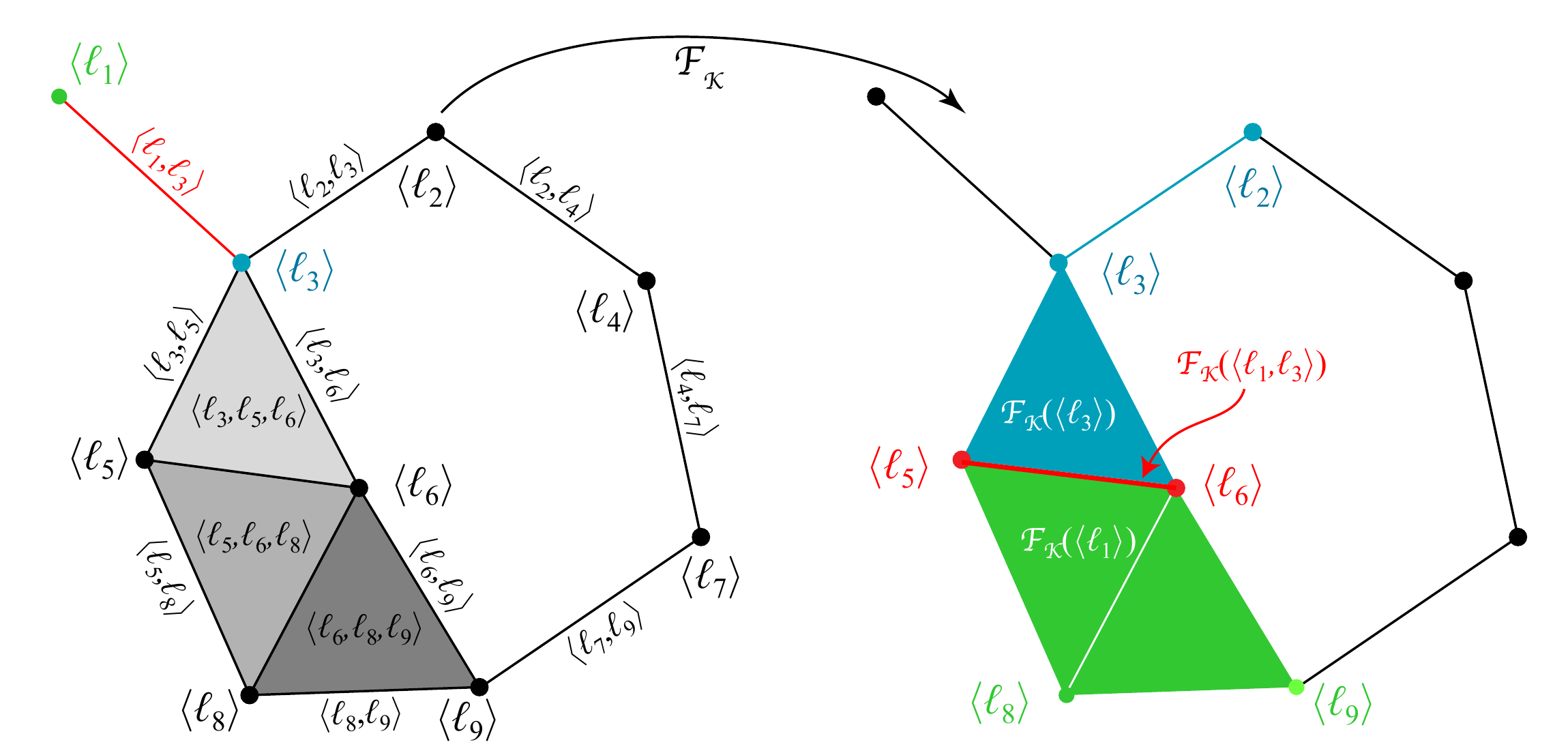}{Simplicial complex $\cK_\alpha(L)$ that is the nerve of the $\alpha$-cells of \Fig{CellularMap}, and the action of the simplicial map $\cFK$ \Eq{FKDefinition} on the simplices $\simplex{l_1}$ (green), $\simplex{l_3}$ (blue) and $\simplex{l_1,l_3}$ (red).}{SimplicialMap}{6in}

A cellular multivalued map $\FA$ induces a combinatorial map $\cFK$ on the nerve that is defined to commute with the correspondence between the grid and its nerve.

\begin{definition}[Simplicial Multivalued Map (SMM)] If $\FA$ is a CMM \Eq{CellularMap} on a grid $\cA$, the map 
$\cFK: \cK \rightrightarrows \cK$, defined by
\beq{FKDefinition} 
	\cFK(\sigma) := \left\{\tau: 
		A_\tau \subset \FA(A_\sigma) \right\} ,
\eeq
is a simplicial multivalued map.
\end{definition}

\noindent
Note that \cFK is a combinatorial multivalued map. Its domain consists
of simplices and its range of \emph{sets} of simplices; moreover, 
$\cFK(\sigma)$ is a sub-complex of $\cK$ (recall \App{Complexes}). As
an example, the simplicial map induced by the cellular map of
\Fig{CellularMap} is shown in \Fig{SimplicialMap}. In this case, since
$A_{\simplex{l_1}} = A_1$, then $
\{\simplex{l_5,l_6,l_8},\simplex{l_6,l_8,l_9}\} \subset
\cFK(\simplex{l_1})$ as are the nine faces of these two
$2$-simplices. Similarly, since $A_{\simplex{l_1,l_3}} = A_1 \cap
A_3$, then $\cFK(\simplex{l_1,l_3}) = \{\simplex{l_5,l_6},
\simplex{l_5},\simplex{l_6}\}$.

The definition \Eq{FKDefinition} satisfies the closed graph condition:
\begin{lemma}[Closed Graph Condition] \label{lem:closed:graph}
	If  $\cFK$ is an SMM and $\tau \le\sigma \in \cK$ (i.e., $\tau$ is a face of $\sigma$), then $\cFK(\tau) \supseteq \cFK(\sigma)$.
\end{lemma}
\begin{proof}
	Since $\tau \le \sigma$, then by \Eq{Atau} $A_\tau \supseteq A_\sigma$. Equation  \Eq{CellularMap} implies that $\FA(A_\tau) \supseteq \FA(A_\sigma)$, and so by \Eq{FKDefinition} $\cFK(\tau) \supseteq \cFK(\sigma)$.
\end{proof}

\noindent

When the nerve $\cK$ of a grid $\cA$ is a geometrical simplicial
complex (again, recall \App{Complexes}), it induces a natural cellular
multivalued map on the geometrical realization $|\cK|$, the union of the
convex hull $|\sigma| \in \bR^n$ of each of its geometrical simplices.  The
map $\FK: |\cK| \rightrightarrows |\cK|$ is the multivalued map
induced by \cFK; i.e., by analogy with \Eq{CellularMap}:
\beq{GeometricSimplicialMap} \FK(x) := \bigcap_{\sigma\in \cK \,:\, x \in
  |\sigma|} \{ |\cFK(\sigma)|\} \eeq

In certain cases, \FA and $\FK$ contain the same information about homology. We can show this when the geometric realization of the complex is a (strong) deformation retract of the geometric realization of the grid, i.e.,  when $|\cK| \subset |\cA|$ and there exists a continuous map $r:|\cA| \times [0,1] \to |\cA|$, such that
\bsplit{DeformationRetract}
		             r(x,0) &= x, \\
		             r(x,t) &= x \mbox{ if }  x \in |\cK|, \mbox{ and}\\
		             r(x,1) &:=\rho(x) \in |\cK| .
\esplit
As we will see in \Sec{AlphaComplex}, this assumption can be verified for an $\alpha$-grid and its nerve. 
We begin with the following ``partial commutativity" lemma.

\begin{lemma} \label{lem:fk:subset:fc}
	Suppose that $\cK = N(\cA)$ is a geometrical simplicial complex. Let \FA and \FK be cellular multivalued maps as in \Eq{CellularMap} and \Eq{GeometricSimplicialMap}. If $|\cK|$ is a deformation retract of $|\cA|$ and $\rho(A_i) = r(A_i,1) \subset A_i$, then for any $x \in |\cA|$, $(\FK \circ \rho)(x) \subseteq (\rho \circ \FA)(x)$.
\end{lemma}
\begin{proof}
Note that since $\rho$ is onto $|\cK|$ and is the identity on $|\cK|$, $\rho(A_i) = |\cK| \cap A_i$ and it suffices to show that for any $x \in |\cA|$, $\FK(\rho(x)) \subseteq \FA(x)$. Furthermore, we note that it follows directly from \Eq{FKDefinition} that for any simplex $\sigma \in \cK$, $|\cFK \left(\sigma \right)| \subset \FA(A_\sigma)$.
For an $x \in A_\sigma$, $\rho(x)$ is a point in a geometric simplex $|\sigma|$.
Therefore, by \Eq{GeometricSimplicialMap},
\[
	\FK(\rho(x)) \subseteq |\cFK( \sigma) |
	\subseteq \FA(A_\sigma) 
\]
as required.
\end{proof}

This result is exactly what is needed to show that
the maps on homology induced by \FK and \FA are isomorphic. In particular, we can prove the following:

\begin{theorem}\label{thm:acyclic} Under the hypotheses of \Lem{fk:subset:fc}, 
whenever $\FA$ is an acyclic multivalued map then $\FK$ induces the same map on homology as $\FA$.
\end{theorem}

\begin{proof}
Whenever $\FA$ is an acyclic multivalued map, then the
definition \Eq{FKDefinition} implies that $\FK$ is as well. Moreover,
since there exists a deformation retract \Eq{DeformationRetract}, $\rho_*$ is the identity; thus, the maps $\rho \circ \FA$ and $\FK \circ \rho$ are also acyclic. Now,
by \Lem{fk:subset:fc}, $\FK \circ \rho$ is a submap of $\rho \circ
\FA$; therefore, it follows that there is a continuous map 
$u:|\cA| \to |\cK|$ that
is a continuous selector for both $\FK \circ \rho$ and $\rho \circ
\FA$. Thus, if $v$ and $w$ are continuous selectors for \FK
and \FA, respectively, then $(v \circ \rho)$ and $(\rho \circ w)$ are
continuous selectors carried by $(\FK \circ \rho)$ and $(\rho \circ \FA)$. 
The acyclic carrier theorem \cite[Thm.~13.3]{Munkres84} then implies that
\[
	(v \circ \rho)_* = u_* = (\rho \circ w)_* \quad
	\Rightarrow \quad v_* \circ \rho_* = \rho_* \circ w_* \quad 
	\Rightarrow \quad v_* = w_* .
\]
\end{proof}

By \Eq{CellularMap}, \FA is semicontinuous and an outer approximation
of $f$; therefore, if \FA is acyclic, it 
induces a well-defined map on the
homology groups such that $(\FA)_* = f_*$.  That is, in order to determine
the Conley index (recall \App{Conley}) 
for an isolated invariant set of $f$, we need only to know the map on homology induced by \FA and check acyclicity. In \App{ComputingConley},
we recall the theoretical framework and algorithms from \cite{Day08} that
can be used to compute the discrete Conley index for a multivalued map.

Note that our construction of the cellular multivalued map, \FA, is still purely theoretical:
since we only know the points in a time series, we do not know
the image of every point in the state space and thus cannot
compute the outer approximation. Without this, there is no direct
method to compute the associated simplicial map $\cFK$ on the nerve.
In \Sec{witness:map} we define a new map, the witness map, that can be
computed algorithmically from time-series data. We then provide a set
of conditions under which the witness map and the CMM \FA contain
the same information about homology. These results allow us to
calculate the Conley index using the witness map rather than \FA. 

\subsection{$\alpha$-Diagrams and Complexes}\label{sec:AlphaComplex}

The concepts in the previous sections can be specialized to the case
of a grid based on an $\alpha$-diagram\footnote
{See \App{Complexes} for a discussion of $\alpha$-diagrams and related simplicial complexes.}
for a finite set of \emph{landmarks}, $L = \{l_1,\ldots,l_{\ell}\}
\subset \bR^n$.  Recalling that $d(x,y)$ is the
Euclidean metric, we denote the closed ball of radius $\alpha$ centered at $l$ by 
\beq{Ball} 
	B_\alpha(l) = \{ x \in \bR^n: d(x,l) \le \alpha\} ,
\eeq
and the distance from $x$ to a set $S \subset \bR^n$ by 
$d(x,S) = \inf_{y\in S} d(x,y)$.  
For a given $\alpha >0$, the $\alpha$-cell
centered at a landmark point $l_i$ is the intersection of
$B_\alpha(l_i)$ with the Voronoi cell of $l_i$:
\beq{AlphaCell} 
	A_i(\alpha) := B_\alpha(l_i) \cap \{x \in \bR^n: d(x,l_i) \le d(x,L) \}.  
\eeq
We denote the collection of
$\alpha$-cells for a set of landmarks $L$ by
$\cA_\alpha(L)$. An example of a set of landmarks and their associated
$\alpha$-cells is shown in \Fig{CellularMap}.

In our application it will be important to choose the landmarks to cover the time series $\Gamma$ with some accuracy; in particular, there should be a landmark within some prescribed distance from each element of $\Gamma$.  The landmarks should also be relatively sparse: the minimal distance between any two should not be too small. These two requirements are linked, as will be discussed in \Sec{witness:map}.  One selection method, advocated by \cite{deSilva04}, is to choose the landmarks from $\Gamma$ itself by a ``max-min" algorithm: choose $L$ recursively by selecting the farthest point in $\Gamma$ from the previous selection until the operative density and sparseness requirements are satisfied, if possible. For the simple example in \Sec{example}, we do not choose the landmarks from the time series, but uniformly in $\bR^n$.  The general problem of finding the most appropriate landmarks in an efficient way is an interesting question for future investigation.

It is straightforward to show that a collection of $\alpha$-cells is a grid:

\begin{lemma} \label{lem:alpha:grid}
	When $\alpha >0$, the set of $\alpha$-cells $\cA_\alpha(L)$ is a grid on $|\cA_\alpha(L)|$.
\end{lemma}

\noindent
This lemma allows us to define a cellular multivalued map \Eq{CellularMap} for any $\alpha$-grid with $\alpha >0$.

The nerve of an $\alpha$-diagram is the \emph{$\alpha$-complex} denoted
\beq{AlphaComplex}
	\cK_\alpha(L) = N(\cA_\alpha(L)).  
\eeq
Since the $\alpha$-complex is a nerve, it is always an abstract simplicial complex, and, as for the general case of \Sec{cellular:maps}, there is an associated simplicial map  $\cFK$ on  \Eq{AlphaComplex}. Moreover, whenever the vertices $l_i$ are in ``general position" (recall \App{Complexes}), $\cK_\alpha(L)$ is a geometrical complex so that its simplices have dimension at most $n$ and their intersections are faces. We will always assume that the landmarks are selected to be in general position.  

For this case, Edelsbrunner proved that the geometrical realization of the nerve $|\cK_\alpha(L)|$ is a deformation retract, \Eq{DeformationRetract}, of $|\cA_\alpha(L)|$ \cite{Edelsbrunner95}.\footnote
{See \Lem{Retract} and the discussion in \App{Complexes}.}
He also showed, since the $\alpha$-cells are convex, that $\rho(\cdot) = r(\cdot,1)$ can be chosen to preserve inclusion in each specific cell: $\rho(A_i) \subset A_i$.  
Thus for an $\alpha$-grid, the hypotheses of \Lem{fk:subset:fc} hold. Consequently, \Thm{acyclic} implies that whenever the cellular map $F_\cA$ on an $\alpha$-grid is an enclosure of a dynamical system $f$, then the Conley index of an isolated invariant set can be computed from the map induced by $F_\cK$ on homology.


\section{The Witness Complex and Map}  \label{sec:witness:map}

In this section, we define a simplicial multivalued \emph{witness map} $\cFW$ on a complex $\cW$ that is a variant of the \emph{witness complexes} introduced by \cite{deSilva04}. 
The goal is to obtain an outer approximation of a continuous map $f: X \to X$ when the only data that we have is a time series $\Gamma = \{x_0,\ldots,x_{T-1}\}$ near an invariant set $\Lambda \subset X$. 
To construct a multivalued map that approximates $f$, we view
the data as ``witnesses" to a set of $\ell$ nearby landmarks, $L = \{l_1,\ldots,l_{\ell}\}$. 
There are many possible methods to choose appropriate landmarks. One strategy, as mentioned in \Sec{AlphaComplex}, is the max-min procedure of \cite{deSilva04}; but \label{page:landmarks} there are many other possible 
methods, and indeed it is not necessary that $L$ be a subset of $\Gamma$.
The landmarks can be viewed either as the centers of the cells of an $\alpha$-grid $\cA_\alpha(L)$, or as the vertices of a witness complex $\cW(\Gamma,L)$.
We show below that if the data satisfy certain density criteria and the landmarks are (more or less) uniformly spaced, then there is an $\alpha$ for which the witness complex is identical to the  $\alpha$-complex $\cK_\alpha(L) = N(\cA_\alpha(L))$.

The temporal ordering of $\Gamma$ and the witness relation give rise to both an SMM $\cFW$ and an associated CMM, $\FW$, on the $\alpha$-grid.
We show that when the map $f$ is Lipschitz and the trajectory $\Gamma$ is dense enough, 
there is an $\alpha$ such that $\FW$ is an outer approximation of $f$.

\subsection{Witness Complex}\label{sec:WitnessComplex}

A witness complex is a simplicial complex based on a finite data set $\Gamma$ that
is intended to be more parsimonious than more traditional Rips or \v{C}ech complexes (recall \App{Complexes}) \cite{deSilva04, deSilva08}. The vertices of the complex are taken from a set of landmarks $L$ whose cardinality is much smaller than that of $\Gamma$. 
The complex $\cW(\Gamma,L)$ consists of those subsets of $L$ that have a \emph{witness} in $\Gamma$. For example, de Silva and Carlsson say a point $x \in \bR^n$ is
a \emph{weak} witness to a $k$-simplex if the $k+1$ nearest
landmarks to $x$ are the vertices of the simplex. If, in addition, $x$ is
equidistant from each of the vertices, then it is a \emph{strong}
witness to the simplex .

Another way to define witness complexes is through a general construction of Dowker \cite{Dowker52} that associates abstract simplicial complexes with a relation; that is, by the selection of a subset
\beq{WitnessRelation}
	R \subset \Gamma \times L.
\eeq
For example, the strong witness complex corresponds to the relation
$
	R_{0} = \{ (x,l) \in \Gamma \times L : d(x,l) = d(x,L)\}
$.
We say that a point $x$ is a witness to a point $l$ if $(x,l) \in R$.  Thus,
\[
	W_R(\Gamma, l) = \{x\in\Gamma: (x,l) \in R\}
\]
is the set of witnesses to the landmark $l$. 
Following Dowker, a relation gives rise to two abstract simplicial complexes, by vertical and  horizontal slices, respectively. 
In our notation, the witness complex is the former: the vertices of each simplex in the complex share a witness:
$
	\bigcap_{l \in \sigma} W_R(\Gamma,l) \neq \emptyset .
$

We will use a more easily computed version of the witness complex,  a \emph{clique} complex, which is the maximal complex with a given set of edges (recall \App{Complexes}).
The clique complex for a given relation \Eq{WitnessRelation} is
\beq{CliqueComplex}
	\cW_R(\Gamma,L) = \{ \sigma: W_R(\Gamma, l) \cap W_R(\Gamma, l') \neq \emptyset,\, \forall l,l' \in \sigma\} .
\eeq
Note that, though the vertices of each edge in $\sigma$ must share a witness, there need not be a common witness to every vertex in $\sigma$.

There are many possible choices for the witness relation $R$. 
\label{page:epsilon}
We choose to use a fuzzy version of the strong-witness relation: 
\beq{FuzzyWitness}
	R_\eps = \{ (x,l) \in \Gamma \times L: d(x,l) \le d(x,L)+ \eps \}.
\eeq
That is, $x$ witnesses all landmarks no more than $\eps$ farther from $x$ than its nearest landmark.\footnote
{In the notation of \cite{deSilva04}, this relation corresponds to the
complex $W(D,\eps,1)$, where $D$ denotes the matrix of distances
between landmarks and witnesses.}
The parameter $\eps$ represents the \emph{fuzziness} of the boundary between cells.
The definition \Eq{FuzzyWitness} becomes the strong witness relation
for $\eps = 0$.
We will denote the set of witnesses to a landmark using \Eq{FuzzyWitness} by 
$
	W_\eps(\Gamma, l), 
$
and the resulting clique complex \Eq{CliqueComplex} by $\cW_\eps(\Gamma,L)$.

An example is shown in \Fig{WitnessRelation} for an orbit of the
logistic map on $[0,1]$ for a parameter value just above the first
period-doubling accumulation point. Here, the landmarks were selected to be every
$30^{th}$ point in the sorted data, an orbit of length $T=300$. The
relation \Eq{FuzzyWitness} for $\eps = 0.01$ is the set of (blue) points near the
diagonal. For the case shown, each point in $\Gamma$ is a witness to at most two landmarks, and so the maximum dimension of a simplex in the complex \Eq{CliqueComplex} is one.

\InsertFig{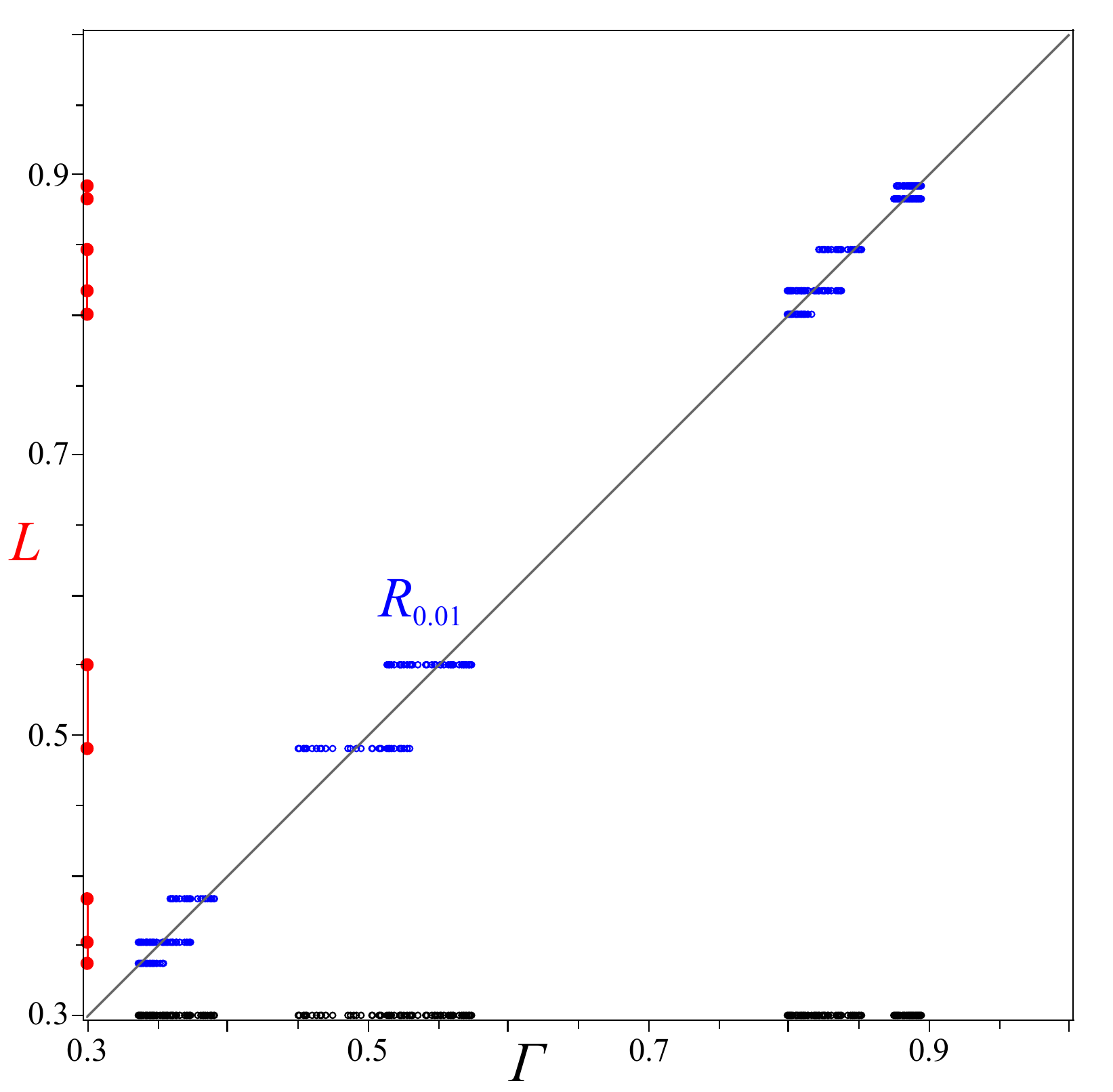}{Witnesses (blue points) defined by the relation \Eq{FuzzyWitness}  for the logistic map, $f(x) = 3.56x(1-x)$ with $\eps = 0.01$. The orbit $\Gamma$, shown along the horizontal axis (black points), has $T= 300$ points, and there are $\ell = 10$ landmarks, shown along the vertical axes (red points). The witness relation defines six one-dimensional simplices (the line segments along the vertical axis), giving a complex with Betti number $\beta_0 = 4$. These correspond to the four major bands in the chaotic attractor of $f$.}{WitnessRelation}{3in}
 	
One way to compute the relation \Eq{FuzzyWitness} is to sort
the rows of the $T \times \ell$ distance matrix $D_{tj} = d(x_t,l_j)$
in order of increasing size; thus, for the sorted matrix $D^s_{t,1} =
d(x_t,L)$. Then $x_t$ is a witness to all of the landmarks in the
first few columns of the $t^{th}$ row of $D^s$, namely those for which
$D^s_{t,j} \le D^s_{t,1} +\eps$.  The main computational expense
here---the distance calculations---can be reduced using a
$kd$-tree \cite{Friedman77}. 
In addition, most implementations of efficient $k$-nearest-neighbor
algorithms return their results sorted in size order.  See
\App{Complexity} for more discussion of algorithms and complexity.

The following section describes how the complex $\cW_\eps(\Gamma,L)$ using the witness
relation \Eq{FuzzyWitness} can be related to an $\alpha$-complex using the same landmark set, under some conditions on the selection of the landmarks and $\alpha$.

\subsection{Equivalence Conditions for $\cK_\alpha(L)$ and $\cW_\eps(\Gamma,L)$}\label{sec:equivalence}

The witness complex is based on the set of landmarks $L$ in $\bR^n$ that can also be viewed as the centers of an $\alpha$-grid $\cA_\alpha(L)$. Since the $\alpha$-diagram limits to 
the Voronoi-diagram, it is clear that for large enough $\alpha$,  
$|\cW_\eps(\Gamma,L)| \subset |\cA_\alpha(L)|$.
Moreover, for large enough $\alpha$, the associated $\alpha$-complex $\cK_\alpha(L)$ is a clique complex---just as we have assumed for the fuzzy witness complex using \Eq{CliqueComplex}. 
We will show here that when this is the case---and if the landmarks are
not too closely spaced---then
$\cW_\eps(\Gamma,L) \subset \cK_\alpha(L)$. Conversely, when the data $\Gamma$ are
dense enough on $|\cA_\alpha(L)|$ we will see that
$\cK_\alpha(L) \subset \cW_\eps(\Gamma,L)$. Consequently, when both sets of conditions are
satisfied, the complexes are the same. As the hypotheses to obtain
these results are independent, we state these two results separately.

\begin{theorem} \label{thm:AlphaSubsetWitness}
For a set of landmarks $L$, a time series $\Gamma$, and $\alpha, \eps > 0$, let $\cK_\alpha(L)$ 
be the $\alpha$-complex \Eq{NerveComplex} and $\cW_\eps(\Gamma,L)$ be the fuzzy witness complex 
\Eq{CliqueComplex} using the relation \Eq{FuzzyWitness}. Suppose that there is a $\delta > 0$ 
such that $\Gamma$ is $\delta$-dense on $|\cA_\alpha(L)|$ and  $\delta \le \eps/2$.
Then $\cK_\alpha(L) \subseteq \cW_\eps(\Gamma,L)$.
\end{theorem}

\begin{proof}
Suppose $\sigma \in \cK_\alpha(L)$, i.e., there is a $y \in |\cA_\alpha(L)|$
such that $\Delta = d(y,L) = d(y,l_i) \le \alpha$ for all $l_i \in \sigma$. We will show that there is an $x \in \Gamma$ that witnesses all the vertices in $\sigma$, i.e., 
\[
	x \in \bigcap_{l\in\sigma} W_\eps(\Gamma,l). 
\]
Since $\Gamma$ is $\delta$-dense, for any $y \in |\cA_\alpha(L)|$, there is at least one point $x \in \Gamma \bigcap B_\delta(y)$. Since $d(y,L) = \Delta$ and $d(x, y) \le \delta$, it follows that 
\[
	d(x,L) \ge \Delta - \delta. 
\]
Since $x \in B_\delta(y)$, for any $l_i \in \sigma$,
\[
	d(x, l_i) \le \Delta + \delta 
	 \le d(x,L) + 2\delta 
	 \le d(x,L) + \eps
\]
since $\delta \leq \eps / 2$. Hence, $x \in W_\eps(\Gamma, l_i)$ for
each vertex of $\sigma$ and therefore, $\sigma \in \cW_\eps$.
\end{proof}

Note that \Thm{AlphaSubsetWitness} applies even when the witness complex is not
defined as a clique complex. However, to show the converse---as we do next---requires
the clique assumption and also relies on the use of the Euclidean metric.

\begin{theorem} \label{thm:WitnessSubsetAlpha}
Suppose $\cK_\alpha(L)$ and $\cW_\eps(\Gamma,L)$ are as in \Thm{AlphaSubsetWitness},
and define
\[
	M = \max_{x \in \Gamma} d(x,L) \quad \mbox{and}\quad 
	\beta = \min_{i\neq j} d(l_i, l_j).
\]
If $\alpha$ is chosen so that so that $\cK_\alpha(L)$ is a clique complex and
\beq{AlphaBound}
	M + \eps \le \alpha \le \tfrac{\beta}{\sqrt2},
\eeq 
then $\cW_\eps(\Gamma,L) \subseteq \cK_\alpha(L)$.
\end{theorem}

\begin{proof}
Note that $\cK_\alpha(L)$ and $\cW_\eps(\Gamma,L)$ have the same vertex set, and by assumption each complex is a clique complex. This means that $\cK_\alpha$ and $\cW_\eps$ are each determined completely by their edges. It follows that we only need to verify that every edge in $\cW_\eps$ is also an edge in $\cK_\alpha$. 

Supposing that $\langle l_1,l_2 \rangle \in
\cW_\eps$, then these landmarks share a witness, i.e., there is 
an $x \in \Gamma$ such that $d(x ,l_i) \le d(x,L) +
\eps$ for $i \in \{1,2\}$. We want to show that there is a point $y \in
|\cA_\alpha(L)|$ such that $d(y,l_1) = d(y,l_2) = \Delta = d(y,L) \le \alpha$. 
In the Euclidean metric there is always a point $y$ equidistant from the two
landmarks such that $d(y,l_i) = \tfrac12 d(l_1,l_2)$. Therefore, since 
$d(l_1,l_2) \le d(x,l_1) + d(x,l_2) \le 2(d(x,L) + \eps)$, then
\[
	\Delta \le d(x,L) + \eps  \le M+\eps \le \alpha,
\]
by \Eq{AlphaBound}. Let $l_3 \in L$ be next closest landmark to $y$, besides $l_1$ and $l_2$, and define $\beta_1 = d(l_3,l_1)$, and $\beta_2 = d(l_3, l_2)$.
As illustrated in \Fig{ProofFig}, $\Delta' = d(l_3, y)$ is
minimized when $\beta_1 = \beta_2 = \beta$. In this case,
the segment from $l_3$ to $y$ is the
perpendicular bisector of the segment from $l_1$ to $l_2$ and so we
have $\Delta \le \Delta'$ only if $\Delta  \le \frac{\beta}{\sqrt{2}}$.
Since $\Delta \le \alpha$, this condition is assured by \Eq{AlphaBound}, and it follows 
that $\langle l_1,l_2\rangle \in \cK_\alpha$. Since $\cK_\alpha$ is a clique complex we have shown that $\cW_\eps \subseteq \cK_\alpha$.
\end{proof}

\InsertFig{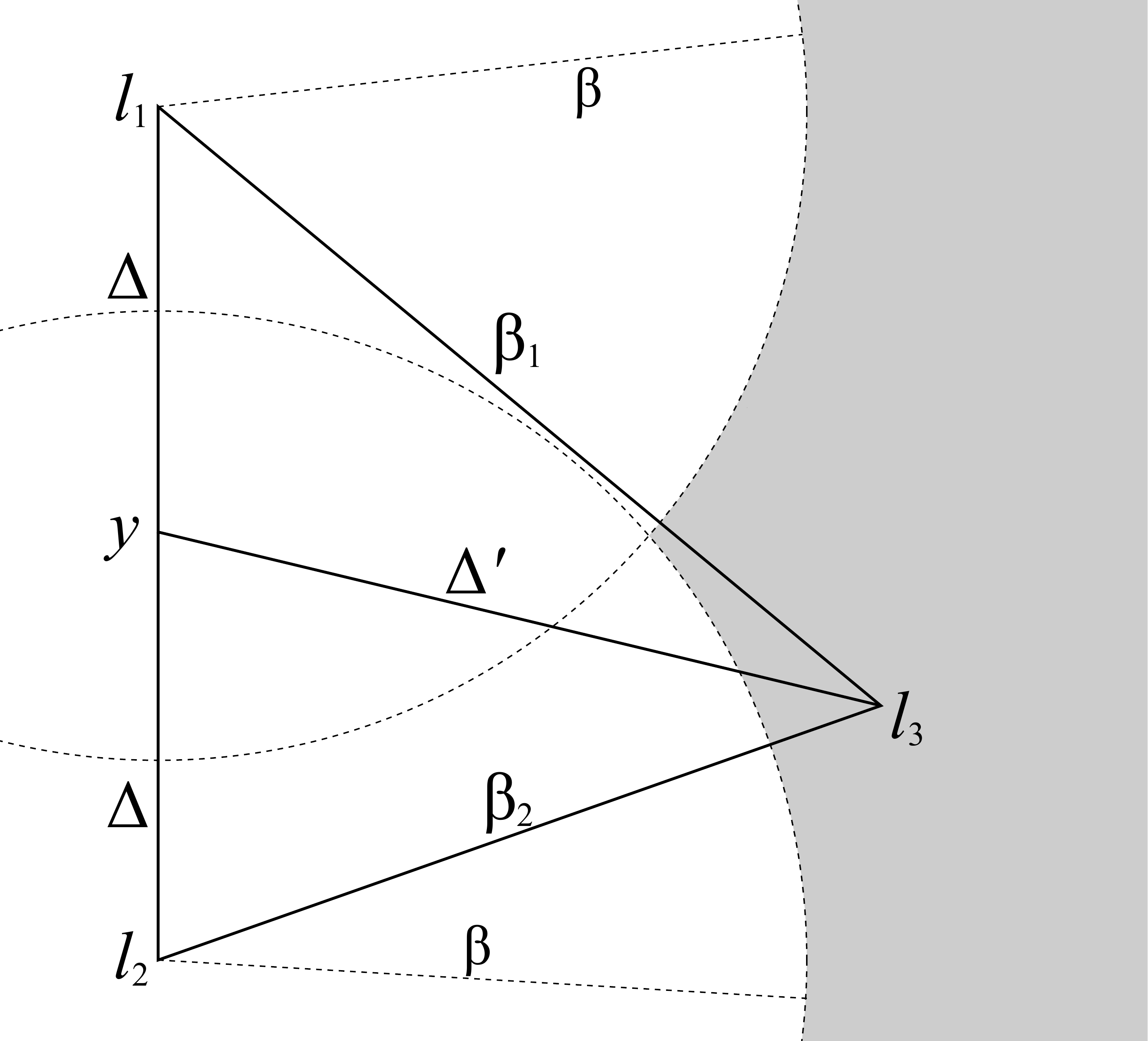}{An illustration of the spacing between three landmarks $l_1,l_2,$ and $l_3$, as in the proof of \Thm{WitnessSubsetAlpha}. The point $y$ is the midpoint between the two landmarks $l_1$ and $l_2$. By assumption, the distance from $l_3$ to $l_1$ or $l_2$ is at least $\beta$, and it is thus minimized when $\beta_1 = \beta_2 = \beta$.}{ProofFig}{3in}

In order to apply \Thm{WitnessSubsetAlpha}, $\cK_\alpha(L)$ must be a
clique complex, which is not always true. However, since the Delaunay complex $\cD(L)$ (recall \App{Complexes}), 
is a clique complex (the Voronoi cells cover $\bR^d$); then $\cK_\infty(L) = \cD(L)$ 
is a clique complex as well. Indeed, whenever $\alpha$ is larger than the radius
of the biggest circumsphere that defines an $n$-dimensional simplex in
$\cD(L)$, then $\cK_\alpha(L) = \cD(L)$. For the simple case of a hexagonal
array of landmarks in $\bR^2$, these circumcircles all have
radius $\beta/\sqrt{3}$, so it is easy to determine when $\cK_\alpha(L)$
is clique. For the trivial case when $\alpha < \beta/2$, the
$\alpha$-balls about each landmark are disjoint, so the
$\alpha$-complex is trivial, and also a clique complex.

When both \Thm{AlphaSubsetWitness} and \Thm{WitnessSubsetAlpha} hold,
then $\cW_\eps(\Gamma,L) = \cK_\alpha(L)$. In this case, a map defined using the
witness relation may have the same homology as a map on $\cA_\alpha(L)$. It is
to this issue that we turn next.

\subsection{Witness Map}\label{sec:WitnessMap}
Abstractly, we can define a cellular multivalued map on a grid $\cA(L)$ that contains the orbit $\Gamma$ using any witness relation: whenever $x \in A_i$ and there is a witness $x_t \in \Gamma$ to the landmark $l_i$, then the image of $x$ should include the cells that $x_{t+1}$ witnesses. The appropriate map is defined similarly to the cellular map
$\FA$, \Eq{CellularMap}, but only using the data $\Gamma$ and the witness relation $R$.

To obtain a cellular map that induces a simplicial map on the witness complex, we assume that the hypotheses of Thms.~\ref{thm:AlphaSubsetWitness}--\ref{thm:WitnessSubsetAlpha} are satisfied. In this case, there are values of $\alpha$ and $\eps$ such that $\cK_\alpha(L) = \cW_\eps(\Gamma,L)$.

\begin{definition}[Cellular Witness Map] Suppose that $\alpha$ and $\eps$ are selected as in Thms.~\ref{thm:AlphaSubsetWitness}--\ref{thm:WitnessSubsetAlpha}. The witness map $\FW: |\cA_\alpha(L)|\rightrightarrows |\cA_\alpha(L)|$ for the fuzzy witness complex $\cW_\eps(\Gamma, L)$ is the cellular multivalued map
\beq{CellularWitnessMap}
	\FW(x) := \bigcap_{A_i \in \cA_\alpha(L) \,:\, x\in A_i} \{ A_j \in \cA_\alpha(L): 
		\, \exists\, x_t \in W_\eps(\Gamma,l_i) \text{ s.t. } x_{t+1} \in W_\eps(\Gamma,l_j) \}.
\eeq
\end{definition}
\noindent
Since by hypothesis the nerve $\cK_\alpha(L) = N(\cA_\alpha(L))$ is also the witness complex, the CMM $\FW$ induces a simplicial multivalued map
\[
	\cFW: \cW_\eps(\Gamma,L) \rightrightarrows \cW_\eps(\Gamma, L)
\]
in precisely the same way that $\cFK$ was induced by $\FA$, namely by \Eq{FKDefinition}.
In other words, a simplex $\tau \in \cFW(\sigma)$ whenever there are witnesses to $\sigma$ that have images, under the temporal ordering of $\Gamma$, that are witnesses to $\tau$.  Indeed, the hypotheses of \Thm{AlphaSubsetWitness} imply that each nonempty image simplex $\tau \in \cFW(\sigma)$, which is automatically in $\cK_\alpha(L)$ since $A_\tau \neq \emptyset$, is also in the witness complex, since $\cK_\alpha(L) \subset \cW_\eps(\Gamma,L)$. Thus to guarantee that the image is in $\cW_\eps(\Gamma,L)$, we need $\Gamma$ to be sufficiently dense on the $\alpha$-shape ($\delta \le \eps/2$). 
Under the additional conditions of \Thm{WitnessSubsetAlpha}, the complexes $\cK_\alpha(L)$ and $\cW_\eps(\Gamma,L)$ coincide and we can view the witness map as having the domain $\cW_\eps(\Gamma,L)$ as well. Thus to guarantee that the domain is well-defined we need that the landmarks are more-or-less uniformly spaced ($\beta$ is not too small), and that each point in $\Gamma$ is not too far from a landmark ($M$ is not too large).

The fact that the $\alpha$ and witness complexes coincide 
gives us hope that $\FW$ will carry the same information about homology as the outer
approximation $\FA$. The following theorem ensures that this indeed is
the case when the original map $f$ satisfies a Lipschitz condition on
the grid.

\begin{theorem} \label{thm:WitnessEnclosure}
Suppose that $Y = |\cA_\alpha(L)|$ is compact and $f$ is Lipschitz on $Y$ with constant $c$. Then if $\Gamma$ is $\delta$-dense on $Y$ and $\delta \le \frac12 \eps \min\{1, \frac1c\}$, $\FW$ is an outer approximation of $f$.
\end{theorem}

\begin{proof}
We need to show that for any $y \in Y$, $f(y) \in \FW(y) $. Note that any such $y \in A_i$ for some $\alpha$-cell $A_i$ and $f(y) \in A_j$ for some other $\alpha$-cell $A_j$. We need to show that $A_j \subset \FW(A_i)$, or specifically, that there is an $x_t \in \Gamma$ such that $x_t \in W_\eps(\Gamma,l_i)$ and $x_{t+1} = f(x_t) \in W_\eps(\Gamma,l_j)$, where $l_i$ and $l_j$ are the landmarks associated with the $\alpha $-cells $A_i$ and $A_j$, respectively. 	 
Since $\Gamma$ is $\delta$-dense, it follows that there is $x \in \Gamma$ with $d(x, y) \le \delta$. Thus, $x$ is at most $\delta$ closer to any landmark than $y$ (whose closest landmark is $l_i$),
\beq{max:dist}
	 	d(x,L) \ge d(y,l_i) - \delta,
\eeq
and consequently:
\beq{is:witness}
		 d(x, l_i) \le d(x, y) + d(y, l_i) \le d(x,L) + 2\delta,
\eeq
Since $2\delta \le \eps$, it follows that $x \in W_\eps(\Gamma, l_i)$. In addition, since $d(f(x), f(y)) \le cd(x , y)$ and $2c\delta \le \eps$, the same reasoning as \Eq{is:witness} leads to $f(x) \in W_\eps(\Gamma,l_j)$.

Note that the points $y$ and $f(y)$ may be in multiple $\alpha$-cells, but the construction above applies to each cell, and so the conclusion  is unaffected.
\end{proof}

We have shown that, under the conditions of
Thms.~\ref{thm:AlphaSubsetWitness}--\ref{thm:WitnessEnclosure}:
\begin{itemize}
\item
the witness complex computed from data has the same homology as the
union of a set of $\alpha$-cells that cover the data, and
\item when
viewed as a multivalued map on $\bR^n$, $\FW$ is an outer approximation
of the dynamical system $f$. 
\end{itemize}
Since the cellular map $\FW$ is semicontinuous
(recall \Sec{cellular:maps}), we know that whenever it is acyclic,
then it is an enclosure of $f$. In this case, the acyclic carrier
theorem implies that 
the induced map on homology can be computed
from any continuous selector to the witness
map \cite[Thm.~13.3]{Munkres84}. However, note that acyclicity cannot
be guaranteed; it must be checked when the map is numerically
constructed.
\label{page:acyclicity}

\subsection{Computing the Map on Homology}\label{sec:Homology}

In our approach, all of the information about the topology of the invariant set $\Lambda \subset X$ 
is contained in the simplicial complex
$\cW = \cW_\eps(\Gamma,L)$, so our computation of the map $f_*$, the action induced by $f$ on the homology groups, relies heavily on this simplicial complex. We begin by recalling
the notion of a chain map.
A chain map from one simplicial complex to another consists of a homomorphism between the vertex sets, a homomorphism between the edge sets, etc., each of which commutes with the boundary operator. Commutation implies, for example, that the boundary of the image of a $k$-simplex is mapped to the image of the boundary of the $k$-simplex. An important feature of a chain map, $\varphi$, is that it induces a well-defined map in homology, $\varphi_*$ \cite{Munkres84}.

Our strategy in calculating $f_*$ is to pick an appropriate chain map,
$\varphi$, so that $f_*$ coincides with $\varphi_*$.  That is, we select $\varphi: \cW \to \cW$ to be a \emph{chain selector}, so that $\varphi(\sigma) \in \cF_{\cW}(\sigma)$ for each $\sigma \in \cK$.
Such a selector can easily be constructed using as a piecewise linear map between the
topological realizations of the simplicial complexes. That is, we
consider $|\varphi|: |\cW| \to |\cW|$. 
If $\varphi$ is a chain selector for the simplicial multivalued map $\cFW$, then it follows that $|\varphi|$ is a continuous selector for $\FW$ and hence, $\varphi_* = f_*$.

In summary, the strategy is as follows. To compute 
the Conley index of an isolated invariant set, 
it is sufficient to construct a cellular multivalued map
$\FA$ that encloses $f$. Yet computing the $\alpha$-grid on a given
data set and its associated cellular multivalued map $\FA$ can be
computationally expensive.  In this section, we have shown how to
construct a sparser simplicial complex---the witness complex
$\cW_\eps(\Gamma, L)$.  In addition, we define two associated
multivalued witness maps - the cellular map, $\FW$, defined implicitly
on $|\cA_\alpha(L)|$ and the combinatorial simplicial map, $\cFW$, defined on
the finitely determined complex $\cW_\eps(\Gamma,L)$. Moreover, when the
hypotheses of the theorems in this section can be verified, then
$\cW_\eps(\Gamma,L) = \cK_\alpha(L)$ and $\FW$ encloses the
dynamical system $f$.  Then any continuous selector for this map will
capture the homomorphism on homology $f_*$.  We get this continuous selector of $\FW$ by constructing a continuous selector of $\cFW$ and taking its geometric realization.


\section{An example: The H\'enon Map}  \label{sec:example}

The procedure for putting the mathematics of the previous sections
into practice on time-series data from a dynamical system is as
follows:

\begin{enumerate}
	\item Given a time series $\Gamma = \{x_0,\ldots,x_{T-1}\}
          \subset X$, which we assume lies near an
          invariant set $\Lambda\subseteq X$, select a set of
          landmarks, $L = \{l_0,\ldots,l_{\ell-1}\}$, that are evenly
          distributed across $\Lambda$ (cf.,
          page~\pageref{page:landmarks} and \App{ComputingConley}).
          These landmarks will be the vertices of a simplicial
          complex.

	\item Choose a value for the $\eps$ parameter that
          satisfies the requirement \Eq{AlphaBound} for
          \Thm{WitnessSubsetAlpha} and use witness/landmark
          relationships to simultaneously define a simplicial complex
          $\cW_\eps(\Gamma, L)$ and a simplicial multivalued
          map, \cFW. Note that, even though we never
          need to construct an $\alpha$-complex, \Eq{AlphaBound}
          implies that there is an $\alpha$ that
          has the same homology as $\cW$.

	\item Pick a subset of $L$ as a starting guess for an isolated
          invariant set and use Algs.~\ref{alg:invariant:part}, 
          \ref{alg:grow:isolating}, and \ref{alg:IndexPair} 
          of \App{ComputingConley} to find an index
          pair $(|N|,|E|)$ for $f$.  There are many different
          strategies for choosing the initial guess; if one is
          attempting to find periodic orbits, for instance, it makes
          sense to search for recurrent points in the time series and
          use the nearest landmark as the starting point for the
          algorithms.  The important property is that the guess should
          be a subset of its period-image under the simplicial multivalued
          map.

	\item Use a chain selector for \cFW to calculate
          $f_*:H_*(|N|,|E|) \to H_*(|N|,|E|)$.

\end{enumerate}

\noindent In the rest of this section, we illustrate this procedure on H\'enon's
classic map \cite{Henon76}:
\beq{HenonMap}
	 f(x, y) = (y + 1 - 1.4x^2, 0.3x).
\eeq
This map has an invariant set $\Lambda$ that is an
attractor, and we generate a trajectory $\Gamma$ of length $T= 10^5$
that starts from the initial condition $z_0 = (-0.4, 0.3)$, near
$\Lambda$.  The trajectory is shown in \Fig{HenonTraj}. As a simple
test of the witness map technique presented in the previous section,
we use this trajectory to verify the trivial fact that $f$ has a fixed
point. We will assume that $\Gamma$ is an exact
trajectory of \Eq{HenonMap}, making no claim that our computation is
rigorous. The latter could, at least in principle, be done using
interval arithmetic.  Given \Eq{HenonMap}, of course, a simple
calculation shows that this system has \emph{two} fixed points.  Our
goal is to find one of those fixed points using only the time series
$\Gamma$.  Note that this is a proof-of-concept example, not an
exhaustive exploration of the parameter space of the algorithm.
Moreover, it is simple enough that the homology calculations can be
carried out by hand.

\InsertFig{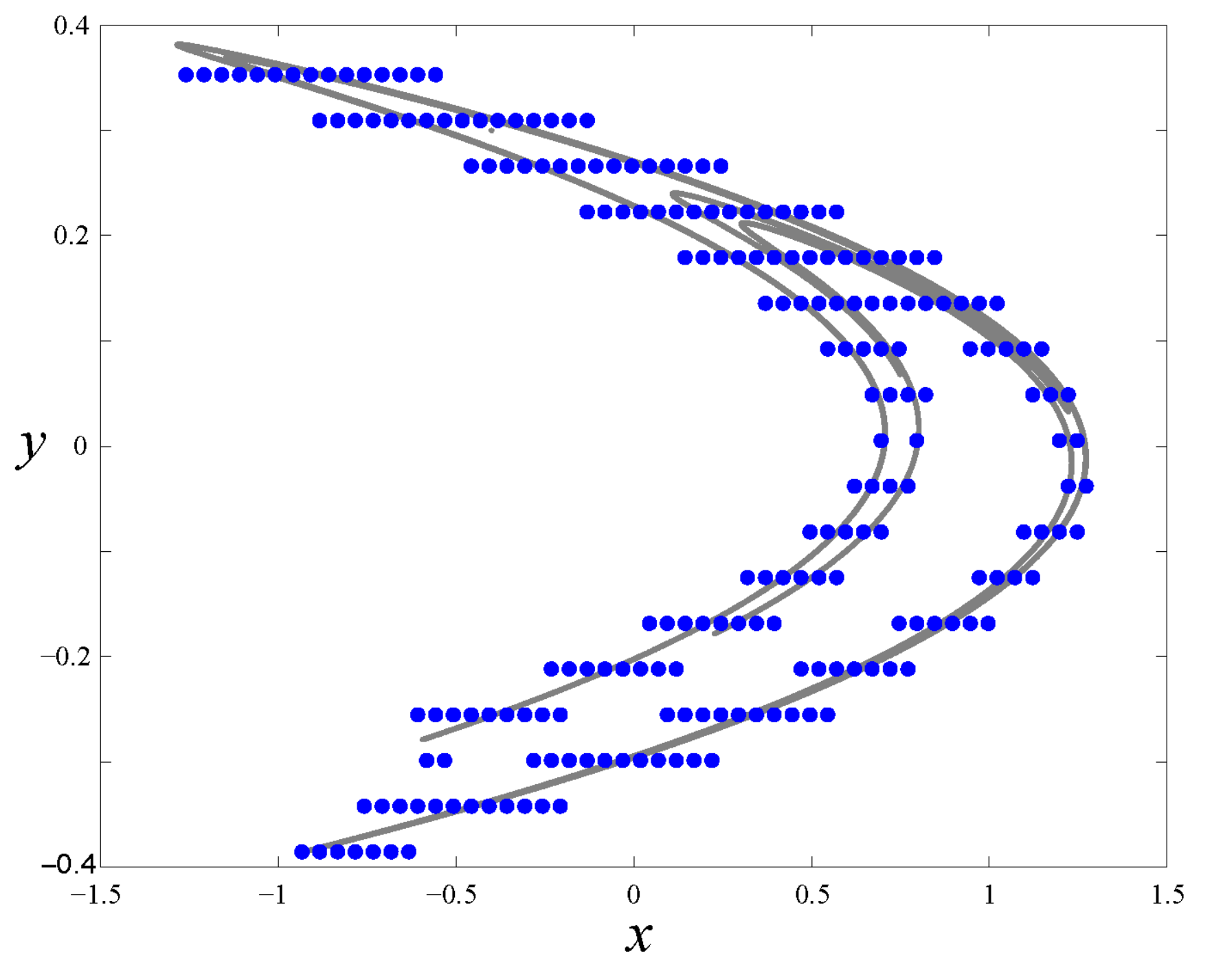}{Time-series data, $\Gamma$, comprising $10^5$ iterates of the
  H\'enon map \Eq{HenonMap} (grey points) from an initial
  condition near the attractor, and the set $L$ of $216$ evenly
  spaced landmarks (blue circles) that approximate the attractor from
  which $\Gamma$ was sampled. Note that the vertical and horizontal
  scales are different.}{HenonTraj}{4in}

We begin by selecting a set of landmarks to approximate the attractor $\Lambda$.
Again, as a proof of principle, we simply space these landmarks evenly
within the bounding box of the orbit $[-1.5,1.5] \times [-0.4, 0.4]$.
With the goal of reflecting the structure of the attractor and yet
having significantly fewer landmarks than points on the orbit ($\ell
\ll T$), we use a hexagonal grid with spacing $\beta = 0.05$. Indeed,
retaining only those landmarks that are within $\beta$ of a
time-series point gives $\ell = 216$ landmarks (so $\ell \sim
\sqrt{T}$), as shown in \Fig{HenonTraj}. This has the effect of
distributing the landmarks across the attractor with enough resolution
to detect some of its fractal structure.

The next step in the process is to define witness-landmark relation
$R_\eps$ of equation \Eq{FuzzyWitness}. 
Although we do not need to build an $\alpha$-complex, it is
possible to do so when the requirement \Eq{AlphaBound} is satisfied. For example, 
given the hexagonal geometry, the
$\alpha$-complex will be a clique complex when $\alpha <
\beta/2$---i.e., when it is trivially totally disconnected---or when $\alpha
\ge \beta/\sqrt{3}$, the distance from a vertex to the center of the
equilateral triangle of side $\beta$. Similarly, by construction,
every data point is in one of the equilateral triangles formed from
the landmarks, i.e., $M \le \frac{\beta}{\sqrt{3}}$. The requirement
\Eq{AlphaBound} is then satisfied if we choose $\alpha = \beta/\sqrt{2}$, and
\[
	\eps \le \left(\tfrac{1}{\sqrt{2}} - \tfrac{1}{\sqrt{3}} \right) \beta.
\]
Recall that given an $\eps \ge 0$, a pair $(x_t,l_j) \in R_\eps \subset \Gamma \times L$ ($x_t \in W_\eps(\Gamma,l_j)$), according to \Eq{FuzzyWitness}, if and only if $d(x_t, l_j) \le d(x_t,L) + \eps$. 
This witness relationship serves to define a
clique complex $\cW_\eps(\Gamma,L)$, by \Eq{CliqueComplex}: the edge $\langle
l_i,l_j\rangle \in \cW_\eps$ if and only if $W_\eps(\Gamma, l_i) \cap W_\eps(\Gamma,l_j) \neq
\emptyset$. By varying $\eps$ up to the bound above and looking at the resulting
complexes, we finally selected $\eps = 0.005$; this is large enough so that the complex is connected, but small enough that the shape of the complex still reflects the primary fold in the attractor. The resulting complex is shown in \Fig{HenonFixedPt}.

Having built the witness complex, we then construct the cellular
multivalued map $\FW: |\cA_\alpha(L)| \rightrightarrows |\cA_\alpha(L)|$ using the
witness-landmark relationships, as described by \Eq{CellularWitnessMap}.  Next
we use the time series to search for an isolating neighborhood for
$\FW$. To apply \Alg{grow:isolating}, we need a guess for an isolating
neighborhood.  For a periodic orbit, this can be found by looking for
nearly recurrent points in the time series
\cite{Lathrop1989}. Consequently, to choose an initial guess for the
purpose of finding a fixed point, we can simply search for a time $t$
that minimizes $d(x_t, x_{t+1})$. Following this approach, we find
that $x_{39,436} = (0.6313, 0.1894)$ is a good candidate---indeed, it
is close to the analytical fixed point $(x^*,y^*) \approx
(0.6313544771, 0.1894063431)$.  
The cell of the landmark nearest to this point gives a useful initial
guess for the isolating neighborhood for \Alg{grow:isolating}.
This isolating neighborhood is then
used as the input of \Alg{IndexPair} to obtain an index pair $(N,E)$
for the fixed point of $f$. The result is shown in \Fig{HenonFixedPt}.

\InsertFig{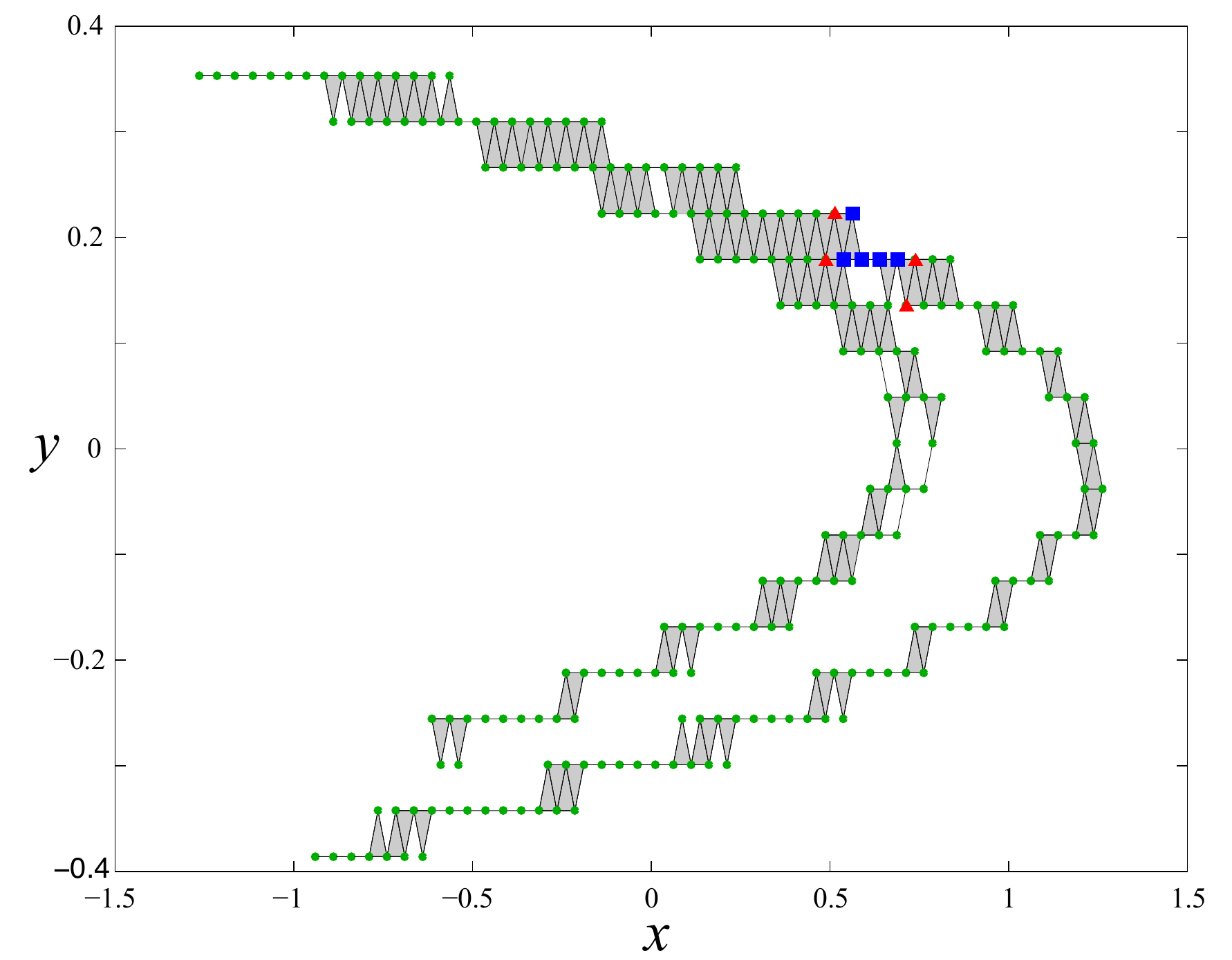}{Witness complex and an index pair $(N,E)$ for
  the fixed point of the H\'enon map. The points correspond to the
  landmarks, $L$, the blue squares represent the isolating
  neighborhood $N \setminus E$ and the red triangles are the
  exit set $E$.  The witness complex corresponds to the landmarks, the
  lines, and the grey triangles.}  {HenonFixedPt}{4in}

The final step is to calculate the Conley index of the isolated
invariant set that is the invariant part of $N \setminus E$.  From
this index, we can then infer the existence of a fixed point.  We
begin by taking a close look at the map \cFW restricted to the index
pair $(N,E)$. The index pair is shown in \Fig{HenonFixedPt}. Recall
that $\FW$ is a map that is constant on $\alpha$-cells. Though we do not
need to compute these $\alpha$-cells, visualizing them---as in the
sketch shown in \Fig{IndexPairGrid}---helps in understanding the
various multivalued maps involved in this process.  In
\Fig{IndexPairGrid}, $N = \{A_1,\ldots,A_9\}$ and $E =
\{A_1,A_2,A_7,A_8\}$. The blue and red landmarks are the nexuses of
the $\alpha$-cells that make up $N\setminus E$ and $E$, respectively.

In this example, the map $\FW$ restricted to the
index pair can be described by the following transition matrix:
\beq{TransitionMatrix} S = \begin{pmatrix} 0 & 0 & 0 & 0 & 0 & 0 & 0 &
  0 & 0 \\ 0 & 0 & 0 & 0 & 0 & 0 & 0 & 0 & 0 \\ 1 & 0 & 0 & 0 & 0 & 1
  & 1 & 0 & 0 \\ 0 & 0 & 0 & 0 & 1 & 1 & 1 & 0 & 0 \\ 0 & 0 & 1 & 1 &
  1 & 0 & 0 & 1 & 1 \\ 0 & 1 & 1 & 0 & 0 & 0 & 0 & 1 & 1 \\ 0 & 0 & 0
  & 0 & 0 & 0 & 0 & 0 & 0 \\ 0 & 0 & 0 & 0 & 0 & 0 & 0 & 0 & 0 \\ 0 &
  0 & 0 & 0 & 0 & 1 & 1 & 0 & 0 \\
		\end{pmatrix},			
\eeq 
where $S_{ij} = 1$ if and only if $A_j \subset
\FW(\Int(A_i))$---i.e., there is a witness of the landmark associated
with $A_i$ whose image under the shift map is a witness of the
landmark associated with $A_j$. For example $\FW(A_6) = 
A_2\cup A_3\cup A_8\cup A_9$.
Geometrically, this image is a disk, and is thus acyclic. Indeed, it
is straightforward---if tedious---to verify from \Eq{TransitionMatrix}
that each cell maps to an acyclic set under the witness map
$\FW|_N$. Moreover, by \Eq{CellularWitnessMap}, whenever $x \in A_i \cap A_j$,
it has an image that is the intersection of the images of the
individual cells. In this way, one can compute the images under $\cFW$
of the twelve one-dimensional simplices in the index pair. These
images, inferred from \Eq{TransitionMatrix}, can be verified to be
acyclic when restricted to $N$. Finally only one of the two-simplices
in $N$ remains in $N$ under the map, namely
$\cF_W(\simplex{l_3,l_4,l_9}) = \simplex{l_6,l_7}$; this image has no
homology.  Consequently, the map $\FW$ is acyclic \emph{on the index
  pair}. This condition is sufficient for our limited purpose here of
verifying that $f$ has a fixed point in $N$.  More generally, the
acyclicity of the witness maps on the entire complex could easily be
automated---as is done for cubical complexes in the software package
``CHomP" \cite{CHOMP}.

\InsertFig{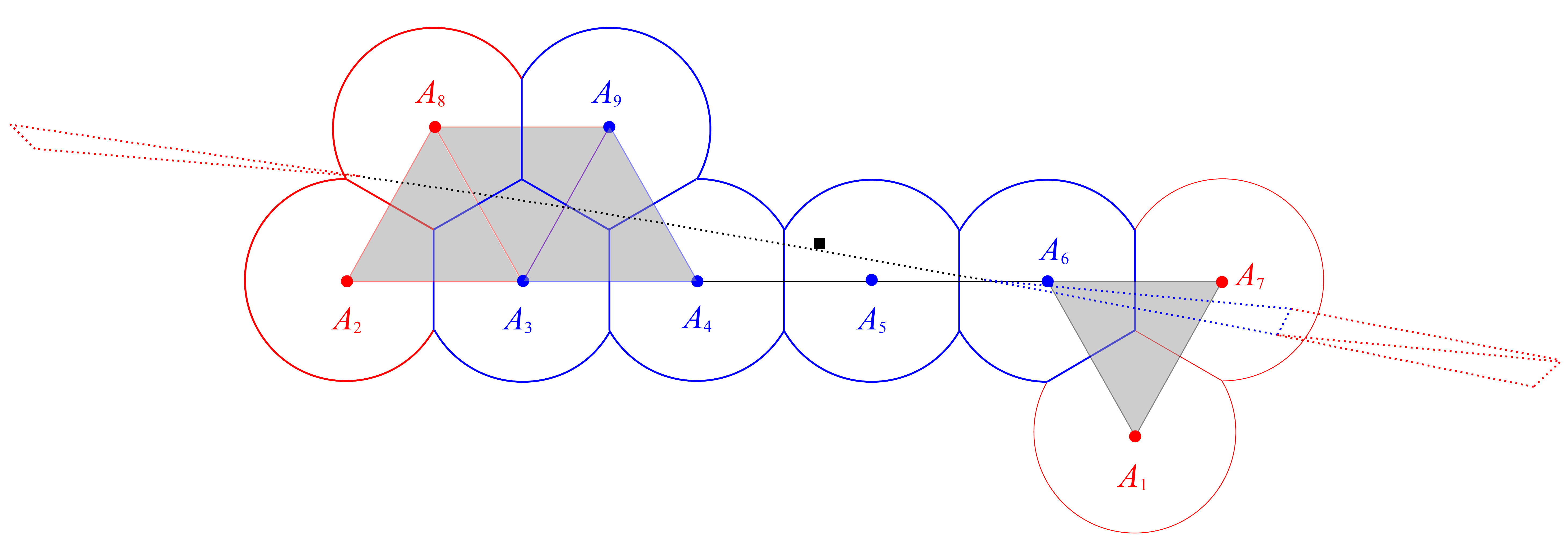}{The witness complex (grey triangles, black
  lines, and dots) for the index pair $(N,E)$ for the landmarks from
  \Fig{HenonFixedPt}.  In this case, the witness complex is also the
  nerve of the $\alpha$-cells (shown as truncated red and blue spheres
  in the figure) since \Eq{AlphaBound} is satisfied.  Dashed lines
  show the analytical image of the simplicial complex under the
  H\'enon map \Eq{HenonMap}. The black square is the fixed point.}
          {IndexPairGrid}{5.5in}

Now, we want to represent this index pair with a corresponding
simplicial complex $\cW(N,E)$---specifically, the witness complex
associated with the landmarks $\{l_1,\ldots,l_9\}$, corresponding to
$\alpha$-cells $\{A_1,\ldots,A_9\}$. The witness complex
$\cW(N,E)$, which is also the $\alpha$-complex in this case,
is pictured in \Fig{IndexPairGrid}.
In order to compute the Conley index, we need a simplicial complex
that represents the quotient space $N/E$. Since the $\alpha$-cells
$A_1, A_2, A_7$, and $A_8$ make up the exit set $E$, we make the
identification
\[
 	l_1 \cong l_7 \cong l_2 \cong l_8 := E.
\]
The resulting simplicial complex is shown in \Fig{QuotientComplex}.
The homology of the simplicial complex $\cK(N,E)$ can easily
be computed by hand in this example. In particular, the quotient space
$N/E$ consists of a single connected component so $H_{*0}(N,E) =
\bZ_2$. The quotient space has a single, nonbounding cycle:
\[
	\sigma = \simplex{E,l_3} + \simplex{l_3,l_4} + \simplex{l_4,l_5} + 
	\simplex{l_5,l_6} + \simplex{l_6,E},
\]
so $H_{*1}(N,E) = \bZ_2$. 
\InsertFig{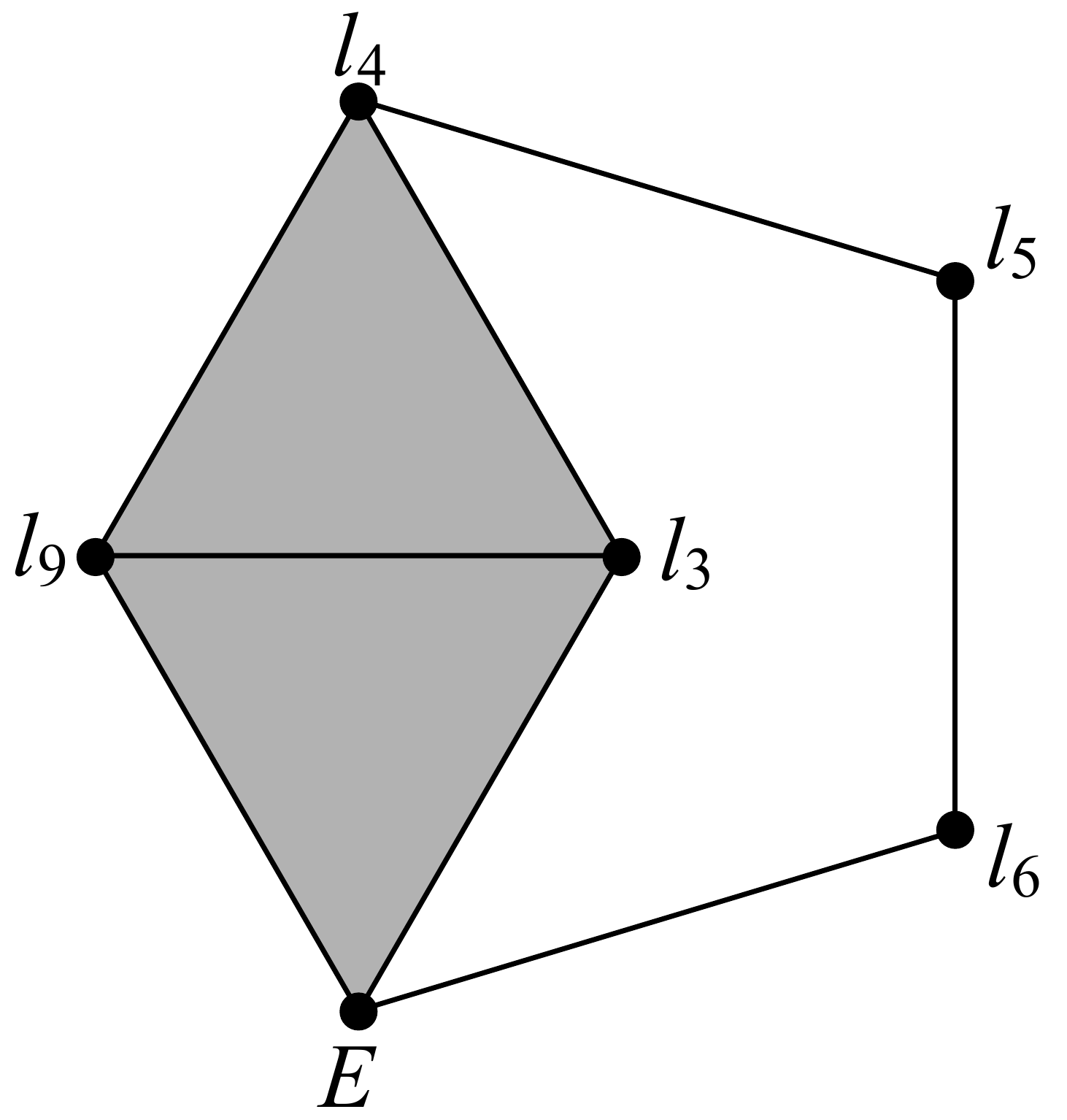}{The quotient simplicial complex
  $N\setminus E$, where $N$ is the simplicial complex shown in
  \Fig{IndexPairGrid} and $E= \{l_1,l_2,l_7,l_8\}$.}{QuotientComplex}{2.0in}

We are now ready to compute the Conley index of the isolating
neighborhood $N \setminus E$. Specifically, we compute $f_*: H_*(N,E)
\to H_*(N,E)$. As described in the previous section, the homology of
$f$ is equivalent to the homology of $\varphi$, where $\varphi$ is a
chain selector for \cFW. Therefore, in order to compute the Conley
index, we need to find $\varphi_*([\sigma]) := [\varphi(\sigma)]$.

The chain selector $\varphi$ is defined inductively by first
determining the image of each vertex in $\cK(N,E)$ (to be
enclosed by \cFW), and then determining the image of each edge so that
$\varphi$ commutes with the boundary operator. In addition, recall
that \cFW on the quotient space must be an enclosure of the index map
$f_N$. We begin with the initial assignment of vertices:
\begin{align*}
	\begin{array}{l |cccccc}
	\simplex{i} & \simplex{E} & \simplex{l_3} & \simplex{l_4} & \simplex{l_5} 
			& \simplex{l_6} & \simplex{l_9} \\
	\hline
	\varphi_0(\simplex{l_i}) & \simplex{E} & \simplex{l_6} & \simplex{l_5} & \simplex{l_3}
			& \simplex{l_9} & \simplex{l_6}
	\end{array}
\end{align*}
To compute $\varphi(\sigma)$, we need to find the images of the edges in $\sigma$ as well.
In order for $\varphi$ to be a chain selector for \cFW, the image of each edge,
$\tau$, must be a subset of $\cFW(\tau)$. Furthermore, $\varphi$ must
commute with the boundary operator, so we need $\varphi_0 \circ
\partial_1 = \partial_1 \circ \varphi_1$. Those two conditions yield
the following edge assignments:
\begin{align*}
	\begin{array}{l|ccccc}
		\tau & \simplex{E,l_3} & \simplex{l_3,l_4} & \simplex{l_4,l_5}
			& \simplex{l_5,l_6} & \simplex{l_6,E}\\
		\hline
		\varphi_1(\tau) & \simplex{E,l_6}& \simplex{l_6,l_5}& \simplex{l_5,l_4}+\simplex{l_4,l_3} 
			&\simplex{l_3,l_9} & \simplex{l_9,E}
	\end{array}
\end{align*}
It follows that $\varphi(\sigma) = \sigma$ (in $\bZ_2$
homology). Since this map is not nilpotent, then it is not in the
shift equivalence class of $[0]$ and thus the invariant set $S =
\inv(N \setminus E,f) \neq \emptyset$ \cite[Thm 10.91]{KMM04}.

With this very simple example, we have illustrated that the witness
complex and the associated witness map can be used to compute the
Conley index for a simple of isolating neighborhood. We plan in
the future to apply this method to more-complex and higher-dimensional
dynamics.


\section{Conclusions \& Future Work} \label{sec:conc}

Computational topology is a powerful way to analyze time-series data
from dynamical systems. Existing approaches to the approximation of a
dynamical system on algebraic objects construct multivalued maps from
the time series using cubical discretizations, then use those maps to
compute the Conley indices for isolated invariant sets of cubes. The
approach described in this paper, by contrast, discretizes the
dynamics using a simplicial complex that is constructed from a
witness-landmark relationship. A natural discretization like this,
whose cell geometry is derived from the data, is more parsimonious and
thus potentially more computationally efficient than a cubical
complex. We then use the temporal ordering of the data to construct a
map on this simplicial witness complex that we call the \emph{witness
map}. Under the conditions established in \Sec{WitnessMap}, this
witness map gives an outer approximation of the dynamics, and thus can
be used to compute the Conley index of isolated invariant sets in the
data.

As a proof of concept, we applied our methods to data from the classic
H\'enon map and located an isolating neighborhood for a fixed point of
this dynamical system. There are many other potential applications in
the study of dynamical systems. Our approach could also be used to
find periodic orbits as well as connecting orbits between them---a
strategy that ultimately leads to rigorous verification of chaotic
dynamics. 

An important question that we leave open is:
can one develop rigorous computational methods for multivalued maps based on
simplicial complexes? While interval arithmetic has been done for
cubical complexes \cite{Day08}, an approach based on the selection of
an appropriate value of $\eps$ to account for finite precision
arithmetic might be appropriate. 

In the future we also hope to explore the
application of these techniques to scalar time-series data sampled
from a dynamical system. In this case, delay-coordinate
embedding \cite{Packard80,Takens81,Sauer91} can be
used to create a trajectory $\Gamma$ of the form required by our
methods. Of course, noise becomes an issue in any consideration of
experimental data. In the case of bounded noise, it may be possible
to turn our $\eps$ parameter to advantage---in the same spirit as in
\cite{Mischaikow99}, where the size of the cells in the multivalued
map is chosen to account for the experimental error.

Our techniques may also have significant impact in the numerical
simulation of differential equations. In \cite{Mischaikow95},
numerical integration---while keeping track of the magnitude of
round-off error---is used to prove that there is chaos present in the
Lorenz equations \cite{Lorenz63}. A key step in this proof is showing
that one can construct a multivalued map that is truly an outer
approximation of a given function $f$. Theorem
\ref{thm:WitnessEnclosure} indicates that our techniques are
appropriate for these types of proofs. The computational efficiency
that the data-driven discretization confers upon the witness-map
construction process should allow this approach to scale well with
dimension, so it is likely that constructions based on this map could
be used to generate computer-based proofs about high-dimensional
differential equations. This would be a significant advance in the
field.

A large body of research in the field of computational topology has
revolved around the concept of
\emph{persistence} \cite{Delfinado95,Edelsbrunner08,Ghrist08,Robins99}. The
idea behind topological persistence is that many computations in this
field depend upon simple scale parameters. For example, the
$\alpha$-complex of a point cloud depends upon the parameter $\alpha$,
the fuzzy witness complex depends upon the parameter $\eps$, etc. It
makes sense, then, to perform these calculations over a range of
parameter values and to search for intervals where the topological
properties remain constant. This was the rationale for the choice of
 $\eps$ in \Sec{example}. A major area for future
research is the development of a theory of persistence in the context
of Conley index theory. We believe that contribution described in this paper is a
significant step in this direction.

\section*{Acknowledgment}
This material is based upon work supported by the National Science
Foundation under Grants \#CMMI-1245947, \#CNS-0720692, and
\#DMS-1211350.  Any opinions, findings, and conclusions or
recommendations expressed in this material are those of the authors
and do not necessarily reflect the views of the National Science
Foundation.

The authors would like to thank the referees for extensive comments
that improved the paper. They also thank Robert Ghrist for pointing
out the Dowker reference and Konstantin Mischaikow and Robert Easton
for useful conversations.


\pagebreak
\section*{Appendices}
\appendix

\section{Simplicial Complexes for Discrete Data}\label{app:Complexes}

An \emph{abstract simplicial complex} $\cK$ is a collection of
finite, ordered subsets $\sigma =\simplex{l_{i_0},\ldots l_{i_k}}$ in
the power set, $2^L$, of a set of vertices $L$ such that if
$\tau \le \sigma$ is a ``face" of $\sigma$ (it contains only vertices also in $\sigma$), 
then it is also a simplex in $\cK$. The vertices, $\simplex{l_i}$, are zero-simplices; edges,
$\simplex{l_i, l_j}$, are one-simplices, etc. The empty set is a face
of every simplex. For points $l_i \in \bR^n$, the geometrical realization of a simplex,
$|\sigma|$, is the convex hull of its vertices. A \emph{geometrical simplicial complex} is an abstract complex $\cK$ such that each intersection $|\sigma| \cap |\tau|$ of simplices in $\cK$ is a face of both \cite{Edelsbrunner95}. 

A simplicial complex $\cL$ is a sub-complex of $\cK$ if every simplex in $\cL$ is in $\cK$. The $k$-skeleton of $\cK$ is the sub-complex containing all simplices of dimension $k$ or less. Thus a one-skeleton is the graph formed from the vertices and edges. 

A \emph{clique} (or \emph{flag}) complex is the maximal complex with a given set of edges \cite{Zomorodian10} (called a ``lazy" complex by \cite{deSilva04}). Thus a clique complex is determined by its one-skeleton.

The \emph{nerve} $N(\cA)$ of a collection of sets $\cA$ is an abstract simplicial complex constructed from finite intersections. The vertices, $l_i \in L$, are labels for the elements $A_i \in \cA$, and a simplex $\sigma =\simplex{l_{0},\ldots, l_{k}}$ is in the nerve if the $k+1$ corresponding sets have nonempty intersection:
\[
	N(\cA) = \{\sigma: \bigcap_{l_i \in \sigma} A_{i} \neq \emptyset \}.
\]
Thus each vertex in the nerve corresponds to the label $l_i$ of the cell $A_i$, and each edge $\simplex{l_i,l_j}$ to a nonempty intersection $A_{i}\cap A_{j}$, etc.
It was shown by Karol Borsuk and Andr\'e Weil that, in certain cases, the nerve has the same homology or homotopy type as the geometrical realization of the collection $\cA$:

\begin{lemma}[Nerve Lemma \cite{Borsuk48, Weil52, Bott82, Hatcher02}]\label{lem:Nerve}
Let $\cA$ be a collection of closed sets such that every finite intersection between its members is either empty or contractible. Then $N(\cA)$ has the same homotopy type as $|\cA|$.
\end{lemma}

There are many natural ways of defining simplicial complexes for a finite point set $L = \{l_1, l_2, \ldots, l_\ell\} \subset \bR^n$. If $\cA = \{B_r(l): l \in L\}$ is the collection of closed radius-$r$ balls \Eq{Ball} around the set of landmarks, then the \emph{\v{C}ech complex}, $\cC_r(L)$, is the nerve of $\cA$. Since the balls are convex subsets of $\bR^n$, the nerve lemma implies that $\cC$ has the same homotopy type as $|\cA|$. The sequence
of \v{C}ech complexes is nested: $\cC_{r}(L) \subseteq \cC_{r'}(L)$ when $r < r'$. A similar complex, the Rips (or Vietoris-Rips) complex, $\cR_r(L)$, consists of all simplices whose vertices are pairwise within a distance $r$ of each other:
\[
 \cR_r(L) = \{\sigma: d(l,l') \le r, \, \forall\, l, l' \in \sigma\}.
\]
Since this complex is determined by its edges, it is a clique complex. Rips complexes are also nested as $r$ grows, and, moreover, they are interleaved with \v{C}ech complexes:
\[
	\cR_{r'}(L) \subset \cC_{r}(L) \subset \cR_{2r}(L)
\]
whenever $r' < r \sqrt{2(n+1)/n}$ \cite{Ghrist08, deSilva07}.
This gives a relation between the persistent homologies of the family of Rips complexes and the family of \v{C}ech complexes.

The \emph{Voronoi diagram} $\cV(L) = \{V_l: l \in L\}$ is the covering of $\bR^n$ by the cells 
\[
	V_l = \{ x \in \bR^n: d(x,l) \le d(x,l'), \, \forall l' \in L\}.
\]
Note that each Voronoi cell is convex, since it is the intersection of half-spaces, and two such cells are either disjoint or they meet on a portion of their boundaries.
The associated simplicial complex is the \emph{Delaunay complex} $\cD(L) = N(\cV(L))$, the nerve of the Voronoi diagram. When the points $L$ are in ``general position" (no more than $n+1$ points lie on any $(n-1)$-sphere) then $\cD(L)$ is a geometrical complex \cite{Edelsbrunner95}. Since this is generically true, general position can be achieved by almost any, arbitrarily small perturbation of the points in $L$. Thus it is common to assume $L$ is in general position.

The cells in the \emph{$\alpha$-diagram} are the intersection of the Voronoi cells with a closed ball of radius $\alpha$ about a vertex:
\[
	\cA_\alpha(L) = \{ V_{l} \cap B_\alpha(l) : l \in L\}
\]
The corresponding nerve is the \emph{$\alpha$-complex}, $\cK_\alpha(L) = N(\cA_\alpha(L))$.
An alternative characterization is that $\sigma \in \cK_\alpha(L)$ if there exists a ball $B_r(x_0)$ with $r \le \alpha$ that contains no vertices in its interior, $\Int(B_r(x_0)) \cap L = \emptyset$, but for which $\sigma \subset \partial B_r(x_0)$. The boundary of such a ball is called a circumsphere for $\sigma$. For the Euclidean metric, each $\alpha$-cell is convex, since it is the intersection of convex sets. Thus the nerve lemma implies that $|\cK_\alpha(L)|$ is homotopy equivalent to $|\cA_\alpha(L)|$.
Note that the $\alpha$-complex is the intersection of the \v{C}ech complex and the Delaunay complex,
\[
	\cK_\alpha(L) = C_\alpha(L) \cap \cD(L),
\]
so it is a sub-complex of each. Moreover, as $\alpha \to \infty$, $\cK_\alpha(L) \to \cD(L)$.

A subset $A \subset X$, is a \emph{deformation retract} of $X$ if there exists a continuous map $r: X \times I \to X$ satisfying \Eq{DeformationRetract}. The restriction $\rho:X \to A$ defined by $\rho(\cdot) = r(\cdot,1)$ is then a \emph{retraction} of $X$ onto $A$. 
If $A$ is a deformation retract of $X$ then it is homotopy equivalent to $X$. More generally,
two spaces $A$ and $B$ are homotopy equivalent iff there is a space $X$ and embeddings $a: A \to X$ and $b: B \to X$ such that both $a(A)$ and $b(B)$ are deformation retracts of $X$ \cite[Cor. O.21]{Hatcher02}. 

In fact the \emph{$\alpha$-shape}, $|\cK_\alpha(L)|$, is a deformation retraction of the $\alpha$-grid $|\cA_\alpha(L)|$ \cite{Edelsbrunner95}. There are two parts to Edelsbrunner's result, and we give only a brief discussion of the ideas in his paper.

\begin{lemma} For any $\alpha \ge 0$ and any finite set of landmarks $L \subset \bR^n$ in general position, $|\cK_\alpha(L)| \subset |\cA_\alpha(L)|$.
\end{lemma}

This follows because when a collection of $\alpha$-cells mutually intersect, i.e., when $A_\sigma \neq 0$, they must do so at a point $x$ in the interior of $|\sigma|$ (viewed
as a subset of its spanning $k$-plane).
The proof proceeds by induction (it is easy for $0$-simplices), and uses the facts that the union $S(\sigma) = \cup_{l\in \sigma} A_l$ is star-convex, relative to $x$, and that $|\sigma|$ is itself convex. The implication is that $|\sigma| \subset S(\sigma)$ for each simplex in $\cK_\alpha(L)$. 
The result is used to construct the deformation retract. 

\begin{lemma}\label{lem:Retract} For any $\alpha \ge 0$ and any finite set of landmarks $L \subset \bR^n$ in general position, $|\cK_\alpha(L)|$ is a deformation retract of $|\cA_\alpha(L)|$.
\end{lemma}

The construction of a deformation retract is based on planes that are orthogonal to points on simplices that are on the boundary of $|\cK_\alpha(L)|$ (for points in the interior, the deformation is the identity map). If $\sigma$ is a $k$-simplex then each point in its interior is the intersection of the $k$-plane containing $|\sigma|$ and the $(n-k)$-plane that is its orthogonal complement. Convexity implies that these families of orthogonal $(n-k)$-planes cover $|\cA_\alpha(L)| \setminus |\cK_\alpha(L)|$, and each point in this set lies in exactly one such plane. The deformation is defined as linear flow from the boundary of $|\cA_\alpha(L)|$ to
the boundary of $|\cK_\alpha(L)|$.
A consequence of this construction is that the deformation maintains membership in each $\alpha$-cell: $r(A_i,t) \subset A_i$. This last property is an hypothesis for \Lem{fk:subset:fc}.

\section{Discrete Conley Index}\label{app:Conley}

A key tool in computational topology is the Conley index \cite{Conley78}, which can be expressed in terms of the algebraic topology of a pair of sets that are acted upon by a map $f$. Given the Conley index of such a pair one can sometimes prove the existence of fixed points, periodic orbits and equivalence to shift dynamics for the dynamics of on invariant set. In this appendix we briefly recall the definition of the index and some related concepts; for more details see \cite{Easton98, Mischaikow02, KMM04}.

Given a homeomorphism $f: X \to X$, a set $\Lambda$ is \emph{invariant} if $f(\Lambda) = \Lambda$. The \emph{maximal invariant set} contained in a set $K$ is 
\[
	\inv(K) = \{x \in K : f^t(x) \in K , \, \forall \, t \in \bZ\},
\]
which, of course, could be empty.
A compact set $K$ is an \emph{isolating neighborhood} if the subset that remains in $K$ for all time is contained in its interior: $\inv(K) \subset \Int(K)$.
Similarly, a set $S$ is an \emph{isolated invariant set} of $f$ when it is the maximal invariant set in the interior of some isolating neighborhood $K$: $S = \inv(K) \subset \Int(K)$.

The computation of the Conley index relies on the construction of an
\emph{index pair} that gives rise to an isolating neighborhood \cite{Mischaikow02}:
\begin{definition}[Index Pair] \label{def:IndexPair}
	A pair of compact sets $P = (N,E)$ with $E \subset N \subset X$ is
	called an \emph{index pair} for $S = \inv(N\setminus E)$ relative to $f$ if it
	satisfies the following three properties.
	\begin{itemize}
		\item $\cl(N\setminus E)$ is an isolating neighborhood.
		\item $f(E) \cap N \subseteq E$.
		\item $f(N\setminus E) \subset N$.
	\end{itemize}
\end{definition}

\noindent
These three properties are illustrated in \Fig{ConleyIndex}.
The first states that that the set $K = \cl(N\setminus
E)$ is an isolating neighborhood that isolates some (possibly empty)
invariant set $S$. The second property implies that once a trajectory
enters $E$, it will not return to $K$ before leaving the index pair
entirely. The third property states that $E$ contains the exit set of
$N$: that is, the images of points not in $E$ must remain in $N$. It is possible to show that every isolated invariant set $S$ has an index pair \cite{Franks00}.

\InsertFig{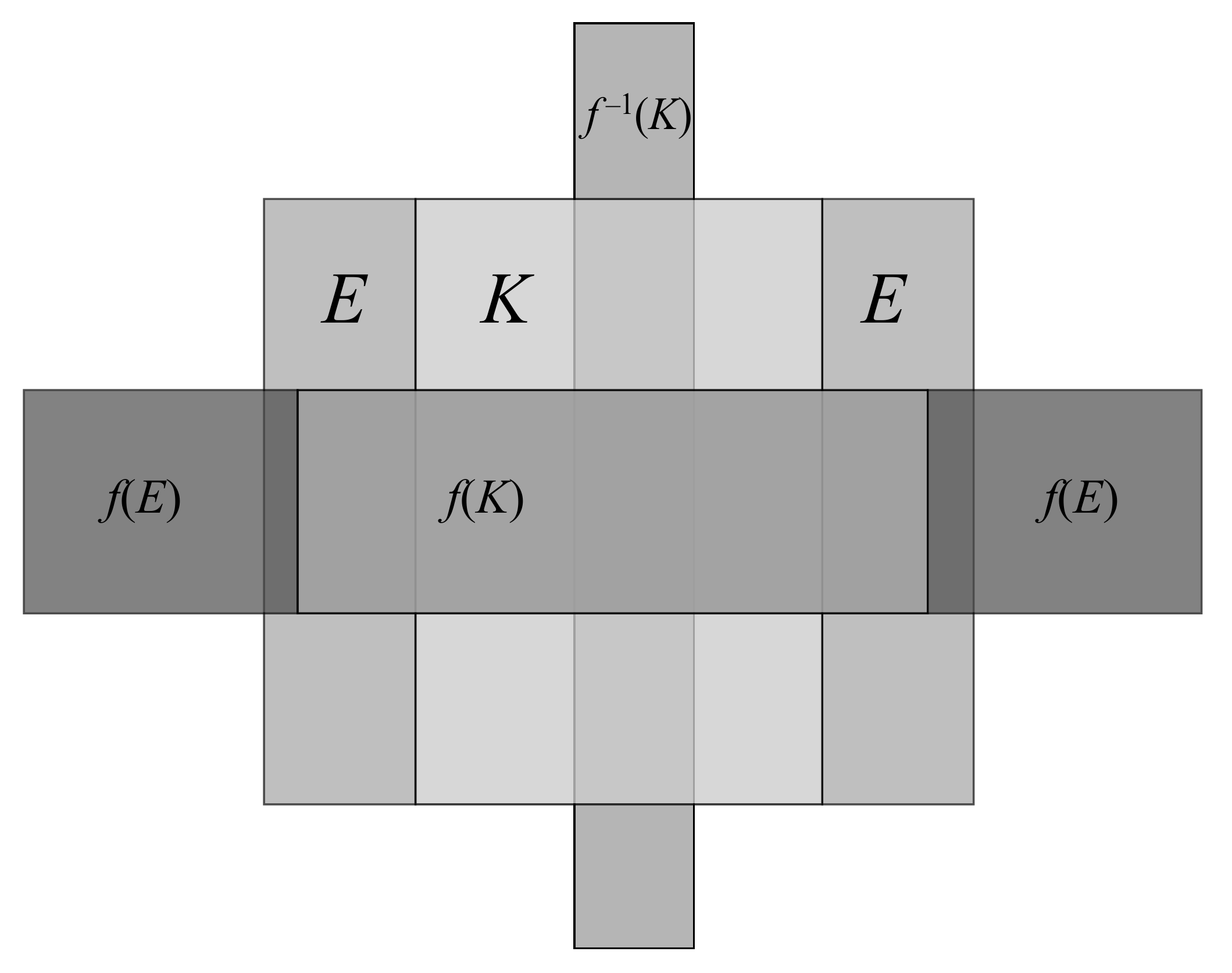}{An index pair $N = K\cup E$ (a square) and $E$ (two rectangles). The image of $E$ either leaves $N$ or remains in $E$. For this picture the isolating neighborhood $K$ is guaranteed to contain a fixed point. A simple map with this index pair is $f(x,y) = (ax, b y)$ with $a>1>b>0$.}{ConleyIndex}{3in}

For any index pair $P = (N,E)$ there is a quotient space $N/E$ with the equivalence relation $[x] = [y]$ if $x,y \in E$.
A continuous \emph{index map}, $f_P$ can be defined on $N/E$ by
\[
	f_P([x]) = \left\{ \begin{array}{ll}
						f(x), & \mbox{if }x,f(x) \in N\setminus E \\
						{[E]}, & \mbox{otherwise}
					\end{array} \right. , 
\]
A complication is that two index pairs $P$, $P'$ for the same isolated invariant set can have topologically distinct quotient spaces and index maps that are not homotopic. 
However, any such index maps induce \emph{shift equivalent} maps on the homology group $H_*(N,E)$ of $N$ relative to $E$. Shift equivalence is a less-rigid version of conjugacy for non-invertible maps that was introduced by Robert Williams and used in the definition of the discrete Conley index by Franks and Richeson \cite{Franks00}.

\begin{definition}[Shift Equivalence] A pair of endomorphisms $f:X \to X$ and $g:Y \to Y$ are \emph{shift equivalent} if there exist continuous maps $h: X \to Y$ and $k: Y \to X$ such that $h \circ f = g \circ h$ and $f \circ k = k \circ g$, and there exists an $m \in \bN$ such that $h \circ k = g^m$ and $k \circ h = f^m$.
\end{definition}

Note that if $f$ and $g$ are homeomorphisms, they are shift equivalence if and only if they are conjugate, with the conjugacy defined by $c = h \circ f^{-m} = g^{-m} \circ h$, and $c^{-1} = k$.

The point is that for any two index pairs $(N,E)$, $(N',E')$ that isolate the same invariant set, the maps on homology $f_{N*}$ and $f_{N'*}$ are shift equivalent \cite{Franks00}.
This shift equivalence class, $[f_{P*}]_s$, is the \emph{discrete homology Conley index} $\Conley(S,f)$ of the invariant set $S$ isolated by $K$.

One of the fundamental advantages of the Conley index is its
structural stability; for example, if $K$ is an isolating neighborhood
for $f$, then there is $\eps > 0$ such that $K$ is also an
isolating neighborhood for $\tilde{f}$, whenever $\|f -
\tilde{f}\|_\infty < \eps$. Moreover, so long as an invariant set
remains isolated by $K$, its Conley index does not change
\cite{Mischaikow02}.

The simplest implication of a nontrivial Conley index is the \emph{Wazewski property}: whenever $\Conley(S,f) \neq [0]$, then $S \neq \emptyset$ \cite[Thm 10.91]{KMM04}. In addition, periodic orbits are guaranteed when the ``Lefschetz number" is nonzero \cite[Thm 10.46]{KMM04}, and (for $C^\infty$ maps) the topological entropy is positive whenever the shift equivalence class of $f_{P}$ has spectral radius greater than one \cite{Baker02}.

\section{Computing the Conley Index}\label{app:ComputingConley}

In order to use the Conley index to obtain information about a map
$f$, we start by using a multivalued map $\FA$ or $\FW$ to locate
isolating neighborhoods for $f$. Since the construction of the CMM is
analogous to the cubical map of
\cite{KMM04}, we borrow the presentation as well as relevant theorems
and algorithms from that work and from \cite{Day08}.
In most cases, the proofs of the theorems in this section are
identical to those in the original citations if one simply substitutes
the concept of a grid-cell for that of a cube. A thorough
treatment of these results---with respect
to any grid satisfying the first definition in
Section~\ref{sec:Grids}---can be found in
\cite{Mrozek99}. Our goal is to move beyond the
cubical complexes used in previous work and devise a method to
\emph{efficiently} build a \emph{simplicial} multivalued map that contains the
same information as the cellular multivalued map.

We begin by defining trajectories and invariant sets for the multivalued map, following 
\cite{Day04,Day08}.

\begin{definition}[Combinatorial Trajectory]
	A \emph{combinatorial trajectory} of \FA through $A \in
 \cA$ is a bi-infinite sequence of cells, $\Gamma_A = (\ldots,
 A^{(-1)},A^{(0)},A^{(1)},\ldots)$, such that $A^{(0)} = A$ and
 $A^{(n+1)} \subseteq \FA(A^{(n)})$ for all $n \in \bZ$.
\end{definition}
 
\begin{definition}[Combinatorial Invariance]
	Given a cellular multivalued map $\FA:|\cA|\rightrightarrows |\cA|$,
 the \emph{combinatorial invariant part} of $N \subset \cA$ is defined by
	\[
		\inv(N,\FA) := \{ A \in \mathcal{A}: \exists\ \text{a trajectory } \Gamma_A 
		\text{ for which }
		A^{(n)} \subset N \text{ for all } n \in \bZ \}
	\]
\end{definition}
\noindent
The following algorithm can be used to locate the combinatorial invariant part of a compact set $N$.
\begin{algorithm} \label{alg:invariant:part}
invariantPart$(N,\FA)$
	\begin{algorithmic}
		\STATE $S \gets N$
		\REPEAT
			\STATE $S' \gets S$
			\STATE $S \gets \FA(S) \cap S \cap \FA^{-1}(S)$
		\UNTIL{$S = S'$}
		\RETURN $S$
	\end{algorithmic}
\end{algorithm}

\noindent
It is proved in \cite[Thm 10.83]{KMM04}---in the context of cubical sets---that if $N$ is finite this algorithm terminates and returns $\inv(N,\FA)$ (which could be empty). The extension to the cellular case is straightforward.

Associated with this notion of invariance, there is a property of isolation, which is defined using:

\begin{definition}[Combinatorial Neighborhood]
	The \emph{combinatorial neighborhood} of a set $S \subset \cA$ is
	\[
		o(S) := \{B \in \cA: B \cap S \neq \emptyset \}.
	\]
\end{definition}
\noindent
More plainly, the combinatorial neighborhood consists of $S$ and all of the cells that touch its boundary. In order for a combinatorial invariant set to be isolated, it must be the invariant set of some neighborhood. 

\begin{definition}[Combinatorial Isolating Neighborhood]
	A set $K \subset \cA$ is a \emph{combinatorial isolating neighborhood} if
	\[
		o(\inv(K,\FA)) \subseteq K
	\]
\end{definition}
\noindent

Given a guess, $K$, for such a neighborhood, we might be able to find an isolating one by \emph{growing} it: if $K' = \inv(o(K),\FA) \subset K$, then $K$ is isolating, otherwise we replace $K$ by $K'$ and repeat. For example, in \Sec{example} where we are looking for a fixed point, we use the cell containing a nearly recurrent point as the initial guess. This leads to the algorithm of \cite{Day04, Day08}:

\begin{algorithm} \label{alg:grow:isolating}
	growIsolating$(K,\FA)$	
	\begin{algorithmic}
		\WHILE{$\inv(o(K),\FA) \not\subset K$}
			\STATE{$K \gets \inv(o(K), \FA)$}
		\IF{$K \cap \partial |\cA| \neq \emptyset$} \RETURN $\emptyset$ \ENDIF
		\ENDWHILE
		\RETURN $K$
	\end{algorithmic}
\end{algorithm}

\noindent If \emph{growIsolating} is called with a combinatorial set
$K \subset \cA$ and a cellular multivalued map $\FA$, then it returns
a combinatorial isolating neighborhood for \FA---or else it fails when
$K$ intersects the boundary of the grid $\cA$. A sufficient condition
for this not to occur is that $|\cA|$ is itself an isolating
neighborhood because then each cell that touches the boundary of
$|\cA|$ has a neighborhood whose invariant part is contained in
$|\cA|$.

An important point is that when $K$ is isolating for $\FA$, then under certain conditions, $|K|$ is isolating for any continuous selector $f$ of $\FA$:
\begin{theorem}
	Let $\FA:|\cA| \rightrightarrows |\cA|$ be a cellular multivalued map for $f$. Then if $K \subset \cA$ is a combinatorial isolating neighborhood for \FA, $|K|$ is an isolating neighborhood for $f$.
\end{theorem}

\noindent This is essentially \cite[Thm~10.87]{KMM04}, generalized to the cellular case.

The computation of the Conley index begins with an isolating neighborhood $K$ of a cellular multivalued map, with the goal of finding a pair of sets $(N,E)$ that satisfy the definition of an index pair. 
We compute these using the following algorithm.

\begin{algorithm} \label{alg:IndexPair}
	indexPair($K,\FA$)
		\begin{algorithmic}
		\STATE $S \gets \inv(K,\FA)$
		\STATE $C \gets o(S)\setminus S$
		\STATE $E \gets \FA(S) \cap C$
		\REPEAT
			\STATE $E' \gets E$
			\STATE $E \gets \FA(E) \cap C \cap E'$
		\UNTIL{$E= E'$}
		\STATE $N \gets S \cup E$
		\RETURN $(N,E)$
	\end{algorithmic}	
\end{algorithm}

\noindent
This is similar to Alg.~10.86 in \cite{KMM04} which was stated for
cubical sets. It was proven there that if this algorithm is called
with a combinatorial isolating neighborhood $K$ and an outer
approximation $\FA$ of $f$, then the geometric realization of the pair
it returns is an index pair for $f$. This proof can be adapted to the
cellular-map situation.

Given an index pair $(|N|,|E|)$ for $f$, the computation of the discrete
Conley index reduces to finding a representative of the shift equivalence
class $[f_{P*}]_s$ and its action on the relative homology groups,
$H_*(|N|,|E|)$. 

\section{Computational Complexity}\label{app:Complexity}

To analyze the computational complexity of the approach proposed in
this paper, and compare it to that of the cubical-grid version
of \cite{KMM04}, one must consider both run time and memory use. 

In the cubical-grid case, all of the cells that are occupied by data
points must be processed in order to compute the homology. The run
time costs of this have two components. Determining whether an
individual data point is in a particular cell in a $d$-dimensional
cubical grid is a matter of evaluating 2$d$ inequalities: the computational cost is $O(d)$. Constructing the multivalued map requires
checking the images of the $2^d$ grid squares that touch the corner
points of the occupied cells and iteratively expanding that set until there are no empty
intersections \cite{Mischaikow99}. 
This iterative expansion step can be a significant computational
expense. Finding an isolated invariant set
can require iteratively checking the
forward and backward images of the neighboring cells in the grid
(cf. \Alg{grow:isolating}). This too can require
significant computational effort.

The witness-complex approach sidesteps all of this complexity in two
ways: first, by using a \emph{subset} of the data; second, by building
a \emph{simplicial} complex from those landmarks. The computational
costs of this approach are balanced differently than in the cubical
grid case: building the complex is harder but using it is easier. In
particular, constructing a witness complex involves calculating the
distances between every point and every landmark, which has cost $O(\ell \log
T)$ if there are $T$ points and $\ell$ landmarks (using, e.g., a $kd$-tree
algorithm). But this computation parallelizes beautifully; moreover, $\ell \ll T$
in practice---indeed, that is the point of the ``coarsening''
inherent in the witness complex. Moreover, the dimension of each simplex is
only high enough to cover the corresponding part of the invariant set,
whereas all of the grid elements in the cubical case necessarily have
the dimension of the ambient space. This means that
\Alg{grow:isolating} not only has fewer cells to process
in the simplicial case, but also far fewer neighbors to check. For
all of these reasons, the overall complexity in computing the homology
of a witness complex is substantially lower than that of the cubical
grid case. Note, too, that the cellular witness map is automatically
an outer approximation if the conditions of
\Thm{WitnessEnclosure} are satisfied.

The memory costs of the two approaches also arise in different ways.
Informally speaking, in order to use a cubical grid to capture the
dynamics with the same fidelity as a witness complex constructed from
landmarks whose minimum spacing is $\beta$, one would need to use grid
elements of size $\beta / \sqrt d $, where $d$ is the dimension of the
ambient space. The number of cells in this grid would be larger than
the number of $d$-dimensional simplices in the corresponding witness complex. This
effect, which holds even if one disregards empty grid cells, may not
be significant in low dimensions and small data sets, but can become
an issue if the data are large and/or high-dimensional. Moreover, if
the landmarks are spaced uniformly in time along the trajectory, that
spacing---and the geometry of the witness complex---naturally adapts
to the dynamics. Cubical grids do not share this advantageous
property.

Another important difference arises in storing the complex in the
computer's memory. There are a number of extremely efficient ways to
store information about which cells of a cubical grid are occupied by
data points. The free-form nature of simplices would appear to make
storing information about them (points, edges, faces, etc.) more of a
challenge, but that cost can be mitigated by using creative algorithms.
Note, for instance, that if one stores the results of the
witness-landmark calculations mentioned above in the form of a linked
list whose $t^{th}$ element contains a list of the landmarks that are
witnessed by the $t^{th}$ data point, sorted in increasing order by
distance, that data structure \emph{contains all of the information one
needs to describe the witness complex}. Algorithmic creativity can
lower the expense of working with that data structure; we are
currently investigating an approach that stores the witness
relationships in bitmap data structures and uses them as ``masks''
(together with logical operations) to find landmarks that are shared
between different sets of witnesses. And the clique assumption made
here can be used to further streamline this search, since all one
needs to consider is the edges.

While we have not provided a test of these claims about computational efficiency on a large set of high-dimensional data in this paper, we plan to do so in future work.

Finally, we would like to note that while building $\alpha$-complexes
is a computationally demanding task in high dimensions,
we never actually \emph{construct} an $\alpha$-complex. The only roles
of that construct in this work are as a vehicle for extending the
proofs of \cite{KMM04} to the simplicial case.

\bibliographystyle{alpha}
\bibliography{Witnessbib}

\end{document}